\documentclass{article}
\usepackage{epstopdf}
\usepackage{amsmath,amsbsy, amsfonts, amssymb, amsthm}
\usepackage{subfigure}
\theoremstyle{plain}
\newtheorem{theorem}{Theorem}[section]

\newtheorem{prop}[theorem]{Proposition}
\newtheorem{algo}[theorem]{Algorithm}

\theoremstyle{definition}

\theoremstyle{remark}
\newtheorem{rem}[theorem]{Remark}
\newtheorem{exemple}[theorem]{Example}
\theoremstyle{plain} 
\usepackage{hyperref}
\usepackage{natbib}
\usepackage{color}
\usepackage{dsfont}
\newcommand{\bx}{\boldsymbol{x}}
\newcommand{\bxc}{\boldsymbol{x}^{\textsc{c}}}
\newcommand{\byc}{\boldsymbol{y}^{\textsc{c}}}
\newcommand{\bxd}{\boldsymbol{x}^{\textsc{d}}}

\newcommand{\by}{\boldsymbol{y}}
\newcommand{\bz}{\boldsymbol{z}}

\newcommand{\bu}{\boldsymbol{u}}
\newcommand{\bt}{\boldsymbol{\theta}}
\newcommand{\bal}{\boldsymbol{\alpha}}
\newcommand{\bpi}{\boldsymbol{\pi}}
\newcommand{\bbet}{\boldsymbol{\beta}}

\newcommand{\bG}{\boldsymbol{\Gamma}}
\newcommand{\bo}{\boldsymbol{0}}

\newcommand{\tx}{\boldsymbol{\mathrm{x}}}
\newcommand{\ty}{\boldsymbol{\mathrm{y}}}
\newcommand{\tz}{\boldsymbol{\mathrm{z}}}

\usepackage{graphicx}
\usepackage{subfigure}

\begin{document}



\title{Model-based clustering of Gaussian copulas for mixed data}

\author{
Matthieu Marbac
and Christophe Biernacki
and Vincent Vandewalle
}

\maketitle

\begin{abstract}
Clustering task of mixed data is a challenging problem. 
In a probabilistic framework, the main difficulty is due to a shortage of conventional distributions for such data.
In this paper, we propose to achieve the mixed data clustering with a Gaussian copula mixture model, since copulas, and in particular the Gaussian ones, are powerful tools for easily modelling the distribution of multivariate variables. 
Indeed, considering a mixing of continuous, integer and ordinal variables (thus all having a cumulative distribution function), this copula mixture model  defines intra-component dependencies similar to a Gaussian mixture, so with classical correlation meaning.
Simultaneously, it preserves standard margins associated to continuous, integer and ordered features, namely the Gaussian, the Poisson and the ordered multinomial distributions.
As an interesting by-product, the proposed mixture model generalizes many well-known ones and also provides tools of visualization based on the parameters.
At a practical level, the Bayesian inference is retained and it is achieved with a Metropolis-within-Gibbs sampler. Experiments on simulated and real data sets finally illustrate the expected advantages of the proposed model for mixed data: flexible and meaningful parametrization combined with visualization features.
\end{abstract}

\textbf{keywords}
Clustering; Gaussian copula; Metropolis-within-Gibbs algorithm; Mixed data; Mixture models; Visualization.

\section{Introduction}
Clustering is an efficient tool for managing large data sets, as it groups individuals into a few specific classes. In a probabilistic framework, a class groups together individuals arising from the same distribution. In such a framework, the most popular approaches model the data distribution using finite mixture models of parametric distributions \citep{Mcl00}. The literature consecrated to homogeneous data (composed of variables of the same type) is extensive and presents Gaussian mixture models \citep{Ban93}, multinomial mixture models  \citep{Goo74} and Poisson mixture models \citep{Kar08} as the standard models used to cluster such data sets. The success of these models is due to the use of conventional distributions for mixture components, thus allowing practitioners to easily interpret the components. However, even though many data sets contain mixed data (variables of different types), only a few multivariate distributions are available for such data sets. Moreover, these distributions can be difficult to interpret. We shall now present the three main models used to cluster mixed data. A more detailed overview is given by \citet{Hun11}.

\medskip
The \emph{locally independent mixture model} \citep{moustaki2005latent} analyzes data sets by assuming that all variables are independent given the component. It can provide meaningful results (see the applications of \citet{Lew98,Han01}), particularly when the one-dimensional marginal distributions of the components follow conventional distributions. However, this model can lead to severe bias when its main assumption is violated (see the application of \citet{Van09}). In such cases, two methods can be considered. The first method consists of selecting a subset of intra-class independent variables \citep{Mau09a}, but some information may be lost. The second method consists of using a model to relax the conditional independence assumption. We shall now present both methods.

\medskip
The \emph{location mixture model} \citep{Krz93,Wil99} assumes that continuous variables follow a multivariate Gaussian distribution conditionally on both component and categorical variables.
This model takes intra-component dependencies into account but requires too many parameters. It was therefore expanded by  \citet{Hun99} by grouping the variables into conditionally independent blocks, where a block comprises no more than one categorical variable and follows a location model. 
However,  the interpretation of components can be complex. Indeed, for a given component, the one-dimensional marginal distribution of the continuous variables are itself  Gaussian mixture models (not conventional distributions!). Moreover, the allocation of variables into blocks is a difficult problem that the authors achieve with an ascending method that is sub-optimal.

\medskip
The \emph{underlying variables mixture model} \citep{Eve88} analyzes data sets with continuous and ordinal variables. It assumes that each discrete variable arises from a latent continuous variable and that all continuous variables (observed and unobserved) follow a Gaussian mixture model. The distribution of observed variables is obtained by integrating each Gaussian component into the subset of latent variables. However, in practice, this computation is not feasible when there are more than two discrete variables. In an effort to study data sets with numerous binary variables,  \citet{Mor12} has expanded this model by estimating the scores of latent variables from those of binary variables. However, the interpretation of components is carried out for score-related parameters (not those related to observed variables).

\medskip
Previous models illustrate the difficulty of managing mixed variables with a model whose interpretation and inference are easy. Moreover, they do not the widespread case of integer variables. The main difficulty is due to a shortage of conventional distributions for mixed data. However, copulas are standard tools to define multivariate distributions in a systematic way, thus having good potentiality for providing a sensible answer.

\medskip
\emph{Copulas} \citep{Joe97,Nel99} can be used to build a multivariate model by defining, on the one hand, the \emph{one-dimensional marginal distributions}, and, on the other, the \emph{dependency model}. Recently, \citet{Smi12} and \citet{Mur13} have modelled the distribution of mixed variables using one Gaussian copula. As pointed-out by \citet{Pit06}, the maximum likelihood inference is very difficult for a Gaussian copula with discrete margins. So, it is often replaced by the \emph{Inference Function for Margins} method performing the inference in two steps (see Chapter 10 of \citet{Joe97}) but which is sub-optimal. When all the variables are continuous, the fixed-point-based algorithm proposed by \citet{Son05} achieves the maximum likelihood estimation, but this approach is not doable for mixed data. Therefore, as shown by \citet{Smi12}, it is more convenient to work in a Bayesian framework since it simplifies the inference by using the latent structure of the model.

\medskip
In this paper, we propose to achieve the mixed data clustering with a \emph{Gaussian copula mixture model}. This new model assumes that each component follow a Gaussian copula \citep{Hof07,Hof11}. Thus, the \emph{intra-component dependencies} are taken into account and are analyzed from the intra-component correlation matrices. Taking advantage of the properties of the copulas, the Gaussian copula mixture model provides \emph{conventional distributions} for all one-dimensional marginal distributions of each component. It results a \emph{three-level framework} which permits a user-friendly interpretation, \emph{i.e.} the proportions indicate the component weights, the one-dimensional marginal parameters of each component roughly describe the classes, and the correlation matrices refine this description. Finally, using the continuous latent structure of the Gaussian copulas, a visualisation tool based on a Principal Component Analysis (PCA) per component is proposed. This visualization provides a summary of main intra-component dependencies and a scatterplot of individuals according to component parameters.

\medskip
This paper is organized as follows. Section~\ref{model} introduces the Gaussian copula mixture model and its links to well-known models. Section~\ref{estim} presents the Metropolis-within-Gibbs algorithm used to perform Bayesian inference. Section~\ref{sim} illustrates algorithm behavior and model robustness through numerical experiments. Section~\ref{app} presents two applications of the new model by clustering two real data sets. Section~\ref{conclusion} concludes the study.

\section{Mixture model of Gaussian copulas} \label{model}
\subsection{Finite mixture model}
The vector $\bx=(x^1,\ldots,x^e) \in \mathbb{R}^c\times\mathcal{X}$ denotes the $e=c+d$ observed variables. Its first $c$ elements are denoted by $\bxc$ and correspond to the subset of the continuous variables defined on the space $\mathbb{R}^c$. Its last $d$ elements are denoted by $\bxd$  and correspond the subset of the discrete variables (integer, ordinal or binary) defined on the space $\mathcal{X}$. Note that if $x^j$ is an ordinal variable with $m_j$ modalities, then it uses a numeric coding $\{1,\ldots,m_j\}$.

Data $\bx$ are assumed to arise from the mixture model of $g$ parametric distributions whose the probability distribution function (pdf) is written as follows
\begin{equation}
p(\bx|\bt)=\sum_{k=1}^g\pi_k p(\bx|\bal_k), \label{mixture}
\end{equation}
where $\bt=(\bpi,\bal)$ denotes the whole parameters. Vector $\bpi=(\pi_1,\ldots,\pi_g)$ is defined on the simplex of size $g$ and groups together the component proportions, where $\pi_k$ is the proportion of component $k$. Vector $\bal=(\bal_1,\ldots,\bal_g)$ groups together the component parameters, where $\bal_k$ denotes the parameters of component $k$.

\subsection{Component modelled by a Gaussian copula}
The Gaussian copula mixture model assumes that each component follows a Gaussian copula.
Thus, component $k$ is parametrized by the correlation matrix $\bG_k$  of size $e\times e$ and by the parameters of the one-dimensional margin distributions $\bbet_k=(\bbet_{k1},\ldots,\bbet_{ke})$ where $\bbet_{kj}$ denotes the parameters of the one-dimensional marginal distribution $j$. 
By grouping all the parameters of component $k$ in $\bal_k=(\bG_k,\bbet_k)$, the cumulative distribution function (cdf) of component $k$ is written as
\begin{equation}
P(\bx|\bal_k)=\Phi_{e}(\Phi_1^{-1}(u^1_k),\ldots,\Phi_1^{-1}(u^e_k)|\bo,\bG_k), \label{cdf}
\end{equation}
where $u_k^j=P(x^j|\bbet_{kj})$ is the value of the cdf of the one-dimensional marginal distribution of variable $j$ for component $k$ evaluated at $x^j$, where $\Phi_e(.|\bo,\bG_k)$ is the cdf of the $e$-variate centred Gaussian distribution with correlation matrix $\bG_k$ and where $\Phi_1^{-1}(.)$  is the inverse cumulative distribution function of the standard univariate Gaussian $\mathcal{N}_1(0,1)$. 

In cluster analysis, the categorical variable $z\in\{1,\ldots,g\}$ which indicates the individual's component membership is unobserved. Moreover, it follows the multinomial distribution $\mathcal{M}_g(\pi_1,\ldots,\pi_g)$. Therefore, mixture models are often interpreted as the marginal distribution of $\bx$ based on the distribution of the variable pair $(\bx,z)$. The Gaussian copula mixture model involves a second latent variable  $\by=(y^1,\ldots,y^e)\in\mathbb{R}^e$ such as $\by|z=k $ follows an $e$-variate centred Gaussian distribution $\mathcal{N}_e(\bo,\bG_k)$. Thus, this model can be interpreted as the marginal distribution of $\bx$ based on the distribution of the variable triplet $(\bx,\by,z)$. Conditionally on $(\by,z=k)$, $\bx$ is defined by
\begin{equation}
x^j=P^{-1}(\Phi_1(y^j)|\bbet_{kj}), \; \forall j=1,\ldots,e. \label{determine}
\end{equation}
Thus, the generative model of the Gaussian copula mixture model is written as
\begin{itemize}
\item Class membership \emph{sampling}: $z \sim \mathcal{M}_g(\pi_1,\ldots,\pi_g)$,
\item Gaussian copula \emph{sampling}: $\by|z=k \sim \mathcal{N}_e(\bo,\bG_k)$,
\item Observed data \emph{deterministic computation}: $\bx$ is obtained from \eqref{determine}.
\end{itemize}

\subsection{Specific distributions for mixed-type variables}
We remind that the first $c$ variables of $\bx$ are continuous and denoted by $\bxc$ while the last $d$ variables are discrete variables $\bxd$ (integer and ordinal). So, the 
pdf of component $k$ can be decomposed as
\begin{equation}
 p(\bx|\bal_k)= p(\bxc|\bal_k) \times p(\bxd|\bxc,\bal_k). \label{decompo}
\end{equation}
The expression of the pdf of component $k$ can be deduced from the cdf of component $k$ defined by \eqref{cdf} only when the one-dimensional marginal distributions of component $k$ (\emph{i.e.} distribution of $x^j|z=k$) are set. For facilitating the model interpretation, we consider that $x^j|z=k$ follows a conventional distribution. More precisely,  $x^j|z=k$ follows a \emph{Gaussian} (respectively \emph{Poisson} and \emph{ordered multinomial}) distribution if variable $x^j$ is continuous (respectively integer and ordinal). We have $\beta_{kj}=(\mu_{kj},\sigma_{kj})$ when $x^j|z=k$ follows the Gaussian distribution of mean $\mu_{kj}$ and variance $\sigma_{kj}^2$, while $\beta_{kj}\in\mathbb{R}^{+*}$ when $x^j$ is integer and $\beta_{kj}$ is defined on the simplex of size $m_j$ when $x^j$ is ordinal.
Even if the marginal distribution of an ordinal variable does not take into account of the order between modalities, this order is crucial since it allows to define the marginal cdf of the variable (it corresponds to what we previously called the ``ordered multinomial''). Obviously, the Gaussian copula mixture model cannot manage data set with categorical variables due to the lack of order between modalities.

Conditionally on $z=k$ and $\beta_{kj}$, the knowledge of variable $x^j$ fully determines the variable $y^j$ when $x^j$ is continuous by $y^j=\frac{x^j-\mu_{kj}}{\sigma_{kj}}$. When $x^j$ is discrete, its knowledge only determines an interval $\mathcal{S}_k^j(x^j)=]b_k^{\ominus}(x^j),b_k^{\oplus}(x^j)]$ where $y^j$ provides the same observed variable $x^j$ by \eqref{determine} with  $b_k^{\ominus}(x^j)=\Phi_1^{-1}(P(x^j-1|\bbet_{kj}))$ and $b_k^{\oplus}(x^j)=\Phi_1^{-1}(P(x^j|\bbet_{kj}))$.
Thus, the pdf of component $k$ is written as follows (in the same order as Equation~\eqref{decompo})
\begin{equation}
 p(\bx|\bal_k)
=\frac{\phi_{c}(\byc|\bo,\bG_{k\textsc{c}\textsc{c}})}{ \prod_{j=1}^{c} \sigma_{kj} }
\times 
\int_{\mathcal{S}_k(\bxd)}\phi_{d}(\bu|\boldsymbol{\mu}_{k}^{\textsc{d}},\boldsymbol{\Sigma}_{k}^{\textsc{d}}) d\bu
, \label{composante}
\end{equation}
where:
\begin{itemize}
\item the correlation matrix
$\boldsymbol{\Gamma}_{k}=\begin{bmatrix}
\boldsymbol{\Gamma}_{k\textsc{c}\textsc{c}} & \boldsymbol{\Gamma}_{k\textsc{c}\textsc{d}} \\ 
\boldsymbol{\Gamma}_{k\textsc{d}\textsc{c}} & \boldsymbol{\Gamma}_{k\textsc{d}\textsc{d}}
\end{bmatrix} $ is decomposed into sub-matrices, for instance
$\boldsymbol{\Gamma}_{k\textsc{c}\textsc{c}}$ is the sub-matrix of $\boldsymbol{\Gamma}_{k}$ composed by the rows and the columns related to the observed continuous variables,
\item $\byc=\big(\frac{x^j-\mu_{kj}}{\sigma_{kj}};j=1,\ldots,c\big)$, $\phi_{c}(\byc|\bo,\bG_{k\textsc{c}\textsc{c}})$ denotes the pdf of $c$-variate centred Gaussian distribution with correlation matrix $\bG_{k\textsc{c}\textsc{c}}$,
\item $\mathcal{S}_k(\bxd)=\mathcal{S}_k^{c+1}(x^{c+1})\times\ldots\times\mathcal{S}_k^{e}(x^e)$, $\phi_{d}(.|\boldsymbol{\mu}_{k}^{\textsc{d}},\boldsymbol{\Sigma}_{k}^{\textsc{d}})$ denotes the pdf of a $d$-variate Gaussian distribution with mean $\boldsymbol{\mu}_{k}^{\textsc{d}}$ and covariance matrix $\boldsymbol{\Sigma}_{k}^{\textsc{d}}$,     $\boldsymbol{\mu}_{k}^{\textsc{d}}=\boldsymbol{\Gamma}_{k\textsc{d}\textsc{c}}\boldsymbol{\Gamma}_{k\textsc{c}\textsc{c}}^{-1}\Psi(\bxc;\bal_k)$ is the conditional mean of $\by^{\textsc{d}}$ 
 and $\boldsymbol{\Sigma}_{k}^{\textsc{d}}=\boldsymbol{\Gamma}_{k\textsc{d}\textsc{d}} - \boldsymbol{\Gamma}_{k\textsc{d}\textsc{c}}\boldsymbol{\Gamma}_{k\textsc{c}\textsc{c}}^{-1}\boldsymbol{\Gamma}_{k\textsc{c}\textsc{d}}$ is its conditional covariance matrix.
\end{itemize}

\begin{rem}[Model identifiability]
The Gaussian copula mixture model is identifiable if, at least, one variable is continuous or integer (see Appendix~\ref{identifiability}).
\end{rem}

\subsection{Strengths of the Gaussian copula mixture model} \label{strengths}
\subsubsection{Related models}
The Gaussian copula mixture model can be used to generalize many conventional mixture models, including the four cases mentioned below.
\begin{itemize}
\item If the correlation matrices are diagonal (\emph{i.e.} $\bG_k=\boldsymbol{I}$, $\forall k=1,\ldots,g$), then the model is equivalent to the locally independent mixture model.
\item If all the variables are continuous (\emph{i.e.} $c=e$ and $d=0$), then the model  is equivalent to the Gaussian mixture model without constraints among parameters \citep{Ban93}. In the spirit of the Gaussian homescedastic model, we also propose a parsimonious version of the  Gaussian copula mixture model by assuming the equality between the correlation matrices. This model is named   \emph{homoscedastic} since the covariance matrices of the latent Gaussian variables are equal between components (\emph{i.e.} $\bG_1=\ldots=\bG_g$). The free correlation model will be now called \emph{heteroscedastic} model).
\item The model is linked to the binned Gaussian mixture model. For example, when variables are ordinal, it is equivalent to the mixture model presented by \citet{Gou06}. In such cases, the model is stable through fusion of modalities. 
\item If the variables are both continuous and ordinal, then the  model is a new parametrization of the model proposed by \citet{Eve88}. It should be noted that Everitt directly estimates the space $\mathcal{S}_k(\bxd)$ containing the antecedents of $\bxd$. Moreover, he uses a simplex algorithm to perform maximum likelihood inference, but this method dramatically limits the number of ordinal variables. The new parametrization of the proposed mixture allows to directly estimate the one-dimensional marginal parameters $\bbet_{kj}$ of each component (see details in Section~\ref{estim}), while the parametrization of Everitt implies the difficult estimation of the space $\mathcal{S}_k(\bxd)$. Thus, the parameter inference is easier.
\end{itemize}

\subsubsection{Standardized coefficient of correlation per class}
The Gaussian copula provides a user-friendly correlation coefficient for each pair of variables. Indeed, when both variables are continuous, it is equal to the upper boundary of the correlation coefficients obtained by monotonic transformation of the variables \citep{Kla97}. Furthermore, when both variables are discrete, it is equal to the polychoric correlation coefficient \citep{Ols79}.

\subsubsection{Data visualization per component: a by-product of Gaussian copulas}
By using the latent vectors of the Gaussian copulas $\by|z$, a PCA-type method allows \emph{visualization} of the individuals \emph{per component} which permits an  identification of main intra-component dependencies. The visualization of component $k$ is performed by computing the coordinates $\mathbb{E}[\by|\bx,z =k;\bal_k]$ and then projecting them onto the PCA region associated with the Gaussian copula of component $k$. This space is obtained directly through spectral decomposition of $\bG_k$. The individuals arising from component $k$ follow a centred Gaussian distribution on this factorial map. Those arising from another component have an expectation not equal to zero. Therefore, an individual located far away from the origin arises from a distribution significantly different from the distribution of component $k$. Finally, the correlation circle summarizes intra-component correlations and avoids direct interpretation of the correlation matrix, which can be fastidious if $e$ is large. The following example illustrates these properties.

\begin{exemple}\label{running_exemple}
Let one continuous, one integer and one binary variables arising, in this order, from the bi-component  Gaussian copula mixture model parametrized by
$$\bpi=(0.5,0.5), \; \bbet_{11}=(-2,1),\; \bbet_{12}=5,\; \bbet_{13}=\bbet_{23}=(0.5,0.5),\;  \bbet_{21}=(2,1),
$$
$$\bbet_{22}=15, \; \bG_{1}=\begin{pmatrix}
1 & -0.4 & 0.4 \\
-0.4 & 1 & 0.4 \\
0.4 & 0.4 & 1\\ \end{pmatrix}\ \text{ and }\bG_{2}=\begin{pmatrix}
1 & 0.8 & 0.1 \\
0.8 & 1 & 0.1 \\
0.1 & 0.1 & 1\\
\end{pmatrix}.$$

\begin{figure}[!ht]
  \centering
  \subfigure[Scatterplot of the individuals described by three variables: one continuous (abscissa), one integer (ordinate) and one binary (symbol). Colors indicate the component memberships]{\label{exemple_native}\includegraphics[scale=0.4]{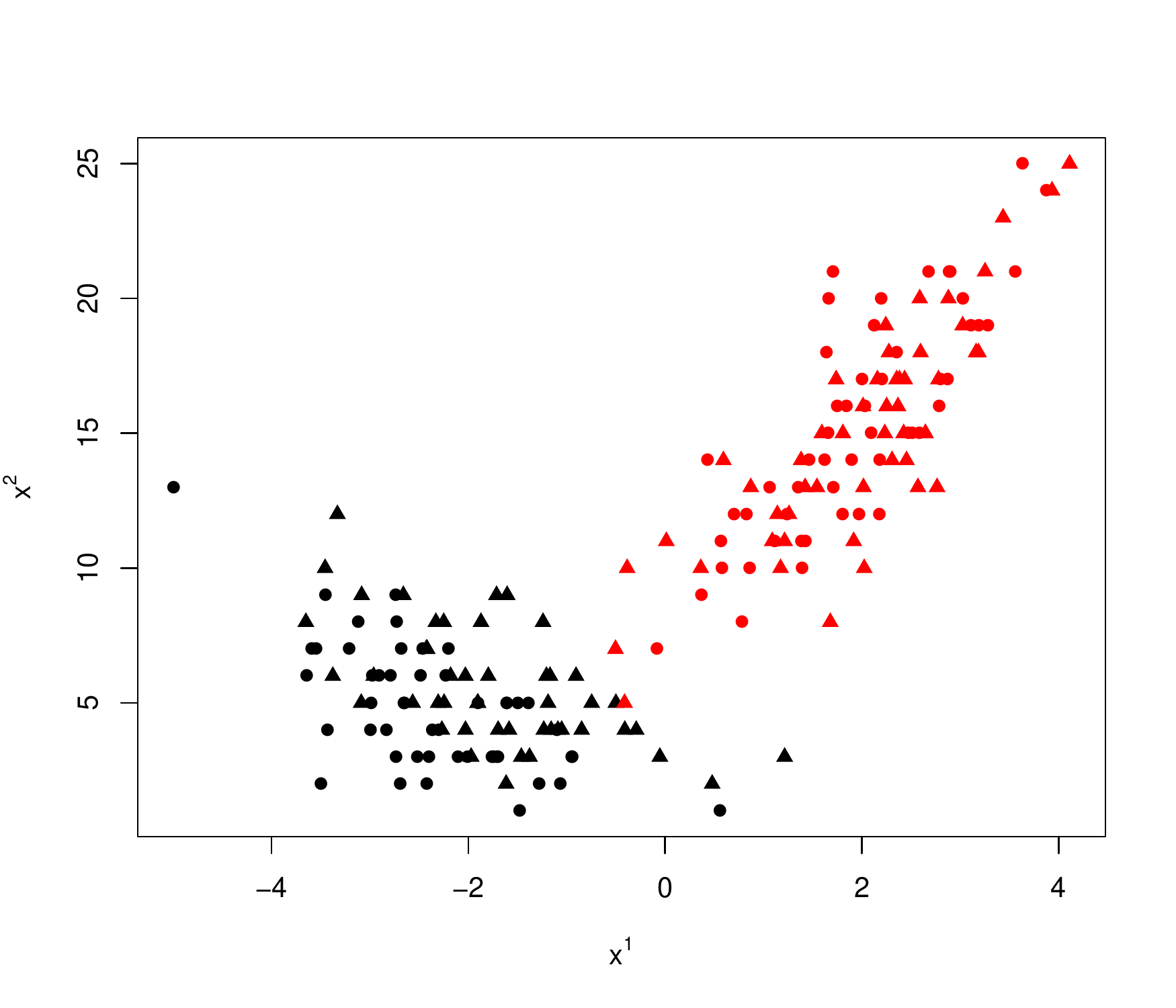} }                \\
  \subfigure[Scatterplot of the individuals in the first factorial map of component~2. Colors and symbols indicate the component memberships]{\label{exemple_acp}\includegraphics[scale=0.4]{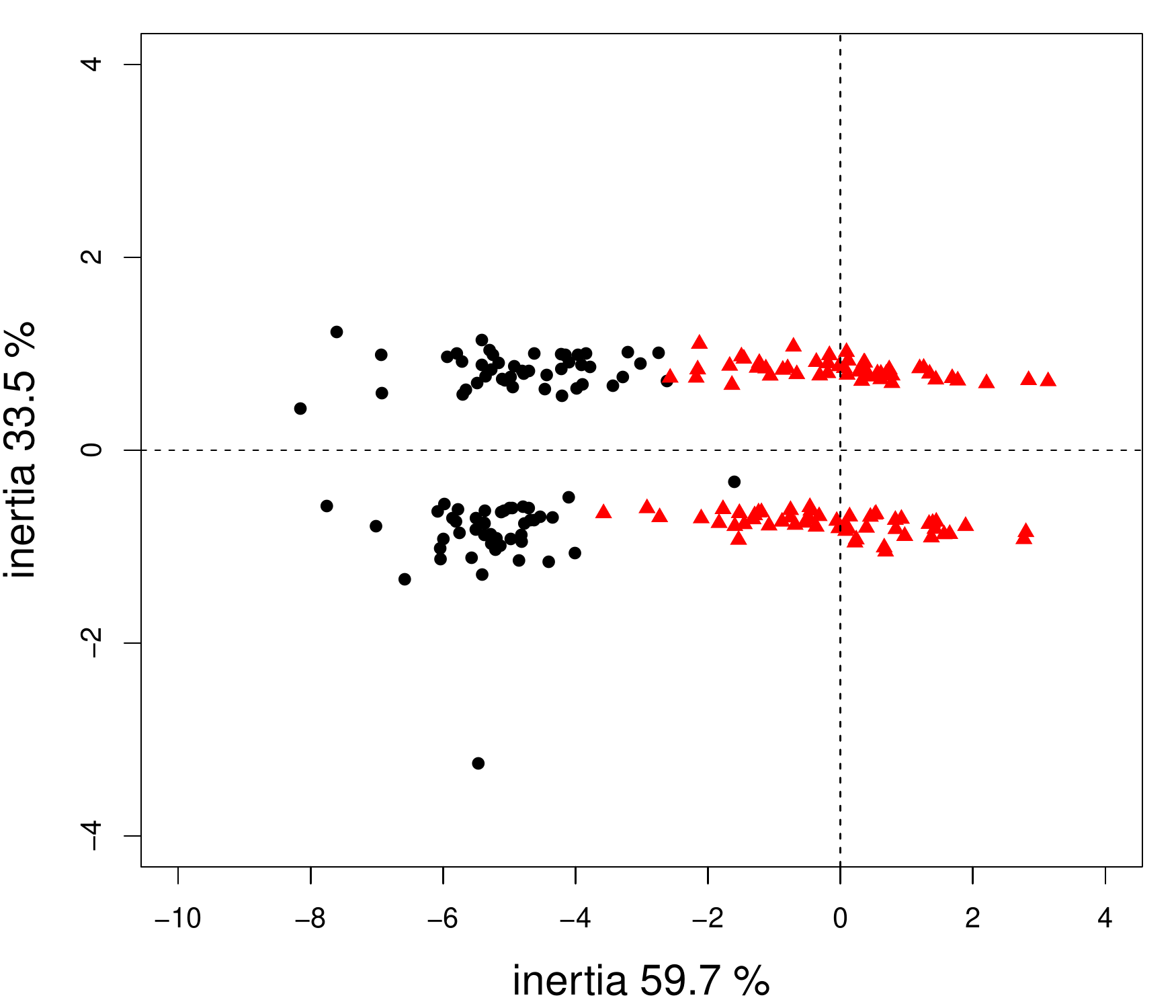} }
  \subfigure[representation of the variables in the first factorial map of component~2]{\label{exemple_acp2}\includegraphics[scale=0.4]{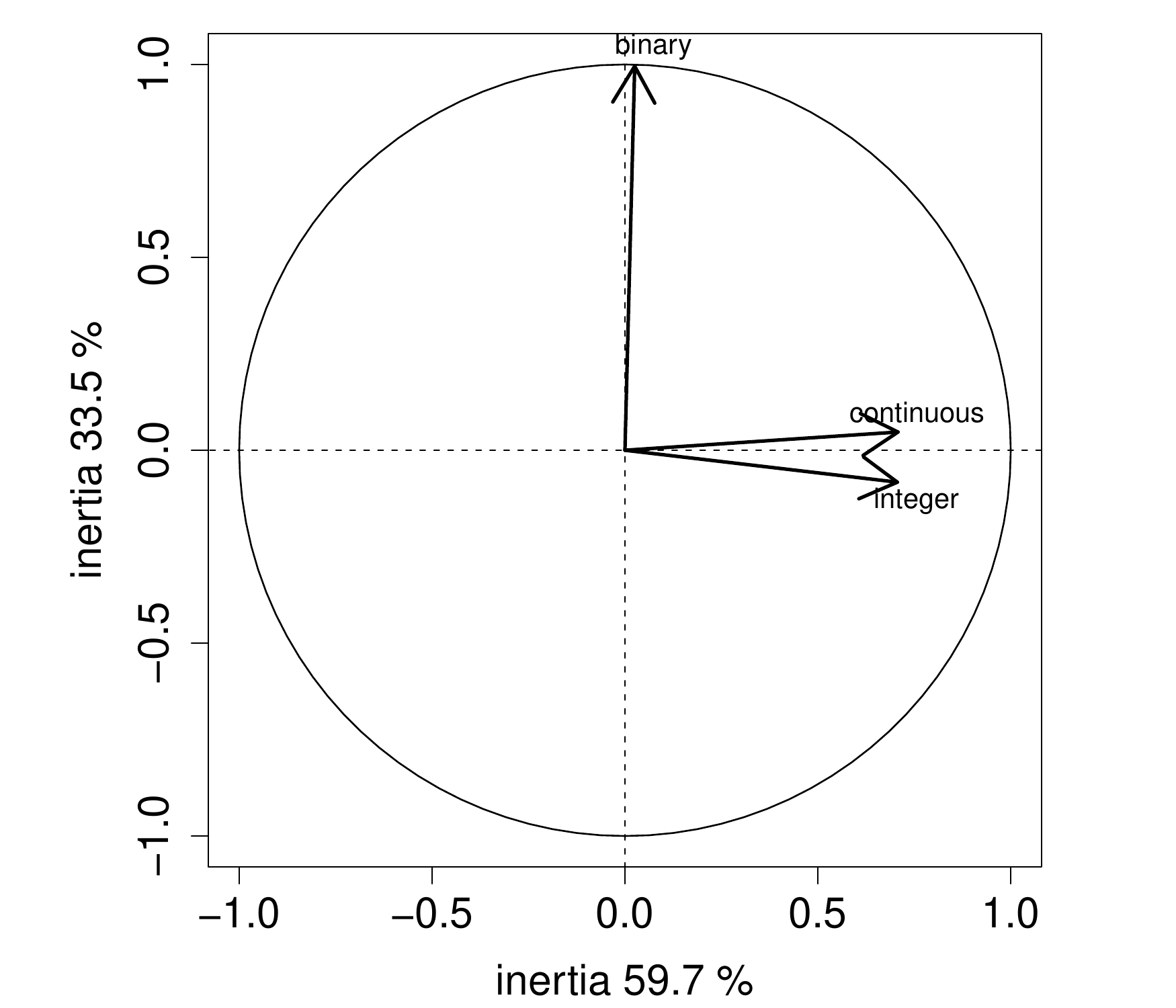} }
  \caption{Example of data visualization.}\label{visu}
\end{figure}
Figure~\ref{visu} provides an example of the data visualization.
Figure~\ref{exemple_native} shows the scatterplot of the individuals in their native space.
Figure~\ref{exemple_acp} presents the scatterplot of the individuals in the first PCA-map of the second component (red). It permits to easily distinguish two classes: a centred one (red) and a second one (black) located on the left side. More precisely, the first axis (explained by the continuous and the integer variables) is strongly discriminative  while the second axis (exclusively explained by the binary variable) is not discriminative. Figure~\ref{exemple_acp2} shows the correlation circle of the first PCA-map of the red component. It allows to point-out a strong correlation, for the red component, between the continuous and the integer variables.
\end{exemple}

\section{Bayesian inference} \label{estim}
\subsection{Sampling layout on data and parameters}
We observe the sample $\tx=(\bx_1,\ldots,\bx_n)$ composed of $n$ independent individuals $\bx_i\in\mathbb{R}^{c}\times \mathcal{X}$ assumed to arise from a Gaussian copula mixture model. As pointed-out by \citet{Smi12}, the Bayesian framework considerably simplifies the inference since it uses the latent structure of the model $(\by,z)$.

Without prior information about the data, we assume independence between the prior distributions.
The proportions and the parameters of the one-dimensional marginal distributions of each component $\bbet_{kj}$ follow the classical conjugate prior distributions \citep{Rob07}. Finally, the conjugate prior of the covariance matrices is derived from an Inverse Wishart distribution as proposed by \citep{Hof07}. Details on the prior distributions are given in Appendix~\ref{prior_margin}.

\subsection{Gibbs and Metropolis-within-Gibbs samplers}
The Bayesian estimation is managed by a Gibbs sampler (described in Algorithm~\ref{Gibbs}) which is the most popular approach to infer mixture models since it uses the latent structure of the data. Its stationary distribution is $p(\bt,\ty,\tz|\tx)$ where $\tz=(z_1,\ldots,z_n)$ denotes the class memberships of $\tx$  and  where $\ty=(\by_1,\ldots,\by_n)$ denotes the Gaussian vector related to $\tx$. Thus, the sequence of the generated parameters is sampled from the marginal posterior distribution $p(\bt|\tx)$. 

Note that the Gaussian variable $\ty$ is twice sampled during one iteration of the algorithm to manage the strong dependencies between $\ty$ and  $\tz$, and between $\by_{[rk]}^j=\{y_{i}^j: z_i^{(r)}=k\}$ and $\bbet_{kj}$. Obviously, the stationary distribution stays unchanged.

\begin{algo}[The Gibbs sampler] \label{Gibbs} Starting from an initial value $\bt^{(0)}$, its iteration $(r)$ consists in the following four steps
($k\in\{1,\ldots,g\}, j \in \{1,\ldots,e\})$ 
\begin{align}
\tz^{(r)},\ty^{(r-1/2)} &\sim \tz,\ty|\tx,\bt^{(r-1)} \label{latent_variables}\\
 \quad \bbet_{kj}^{(r)},\by_{[rk]}^{j(r)} &\sim \bbet_{kj},\by_{[rk]}^j|\tx,\by^{\bar{\jmath}(r)}_{[rk]},\tz^{(r)},\bbet_{k\bar{\jmath}}^{(r)},\bG_k^{(r-1)} \label{margins}\\
\bpi^{(r)} &\sim \bpi|\tz^{(r)} \label{proportions}\\
\bG_k^{(r)} &\sim \bG_k | \ty^{(r)},\tz^{(r)} \label{correlations}
\end{align}
where $\by_{[rk]}=\by_{\{i:z_i^{(r)}=k\}}$, $\by^{\bar{\jmath}(r)}_i=(y^{1(r)}_i,\ldots,y^{j-1(r)}_i,y^{j+1(r-1/2)}_i,\ldots,y^{e(r-1/2)}_i)$ and
$\bbet_{k\bar{\jmath}}^{(r)}=(\bbet_{k1}^{(r)},\ldots,$ $\bbet_{kj-1}^{(r)},\bbet_{kj+1}^{(r-1)},\ldots,\bbet_{ke}^{(r-1)})$. 
\end{algo}

The sampling according to \eqref{proportions} and \eqref{correlations} are classical. However, both samplings from \eqref{latent_variables} and \eqref{margins} are difficulty performed. Therefore, they are replaced by one iteration of a Metropolis-Hastings algorithm that does not change the stationary distribution. The resulting algorithm is a Metropolis-within-Gibbs sampler whose the properties are detailed by \citet{Rob04}.
Details about the four steps of the Metropolis-within-Gibbs sampler are given in Appendix~\ref{Metropolis-within-Gibbs sampler}.

By taking the mean of the parameters generated by the Metropolis-within-Gibbs sampler, we obtain a consistent estimate of $\bt$.

\subsection{Label switching problem}
The label switching problem is generally solved by specific procedures \citep{Ste00b}. However, based on the argument of \citet{Jac14}, these techniques are principally impacting when $g$ is known.

However, when the model is used to cluster, the number of classes is unknown, and the model selection is performed by the BIC criterion \citep{Sch78} which simultaneously avoids the label switching phenomenon. Indeed, on the one hand, this criterion selects quite separated classes when the sample size is small, so the label switching is not present (with high probability) in practice because of the class separability. On the other hand, even if it can select more classes when the sample size increases, the label switching problem does not occurred since this phenomenon vanishes asymptotically.

Obviously, when the number of classes is fixed and the size of sample is small, the label switching problem can occur. In such a case, our advice is naturally to use the procedures of \citet{Ste00b}.

\section{Simulations} \label{sim}
In this section, two simulations are used to illustrate the new model. The first simulation shows the relevance of the estimation procedure.  The second simulation illustrates the robustness of the proposed model by analyzing data sampled from a mixture of Poisson distributions.

\paragraph*{Experiment conditions} 
For each situation, 100 samples are generated. Parameters are estimated by taking the mean of the parameters sampled by $10^3$ iteration of Algorithm~\ref{Gibbs} after a burn-in period of $10^2$ iterations. Algorithm~\ref{Gibbs} is initialized with the maximum likelihood estimator of the locally independent model (particularly relevant when intra-class dependencies are small). The Kullback-Leibler divergence \citep{kullback1951information} is used to compare the estimated distribution and the distribution used to sample the data. This divergence is approximated via $10^4$ iterations of a Monte Carlo simulation.

\subsection{Estimation efficiency}
Data sets are composed of one continuous variable, one integer variable and one binary variable. They are sampled from the distribution described in the example of Section~\ref{strengths}.
Results are presented by Figure~\ref{simul3}
According to Figure~\ref{simul3kul}, the estimated distribution converges to the true distribution when the sample size increases. Indeed, the Kullback-Leibler divergence of the estimated model from the true model decreases as a function of sample size and converges to zero. Moreover, as shown by Figure~\ref{simul3error}, the misclassification error rate converges to the theoretical misclassification error rate (equal to 0.005) when $n$ increases. This simulation illustrates the convergence of the estimator computed by averaging the parameters sampled using the algorithm. Finally, the estimation procedure is not too much time consuming since samples of size $n=100$ and $n=1600$ require respectively $15$ and $64$ seconds, on an Intel Core i5-3320M processor.

\begin{figure}[!ht]
  \centering
  \subfigure[Kullback-Leibler divergence of the estimated model from the true model.]{\includegraphics[scale=0.4]{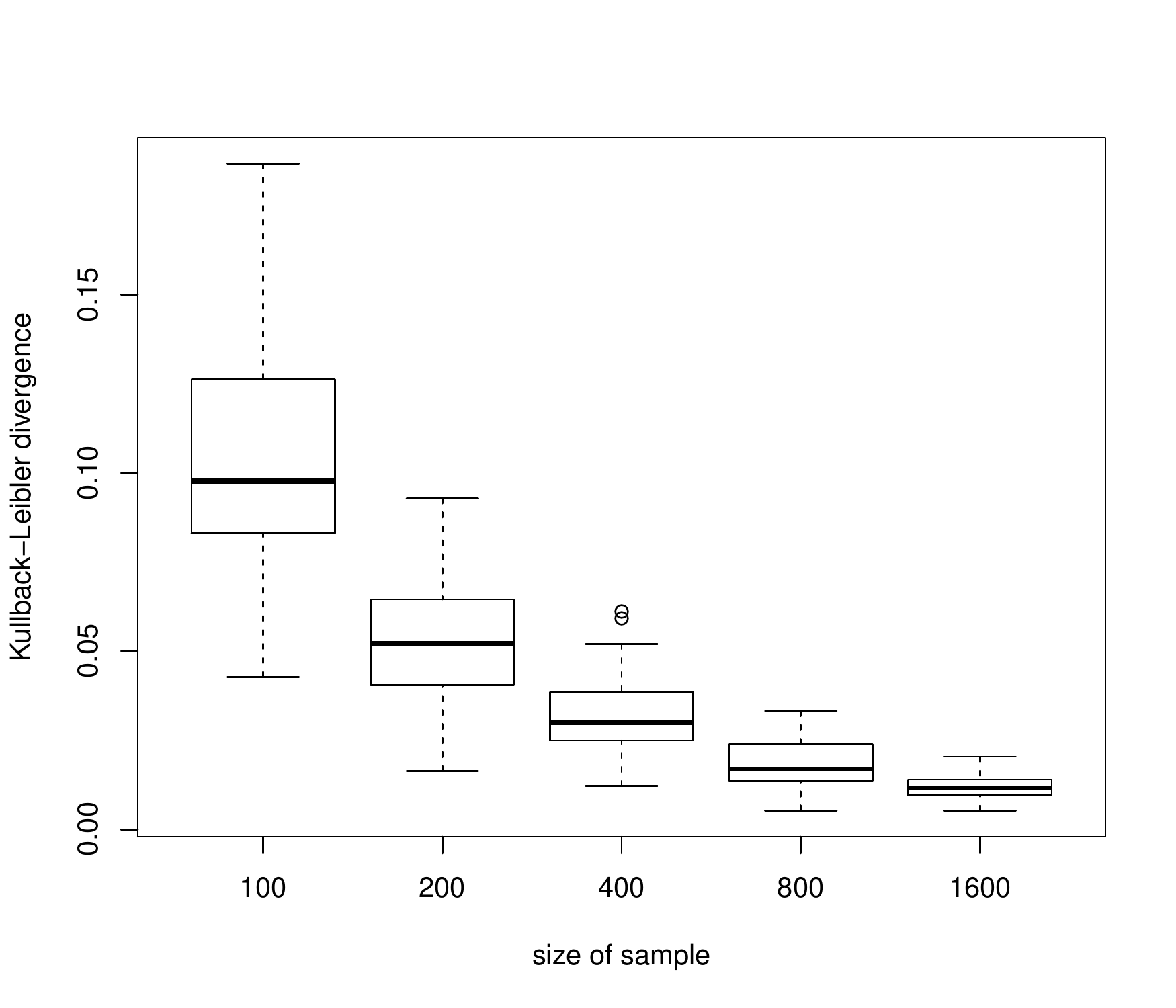} \label{simul3kul}}
  \subfigure[Misclassification error rate.]{\includegraphics[scale=0.4]{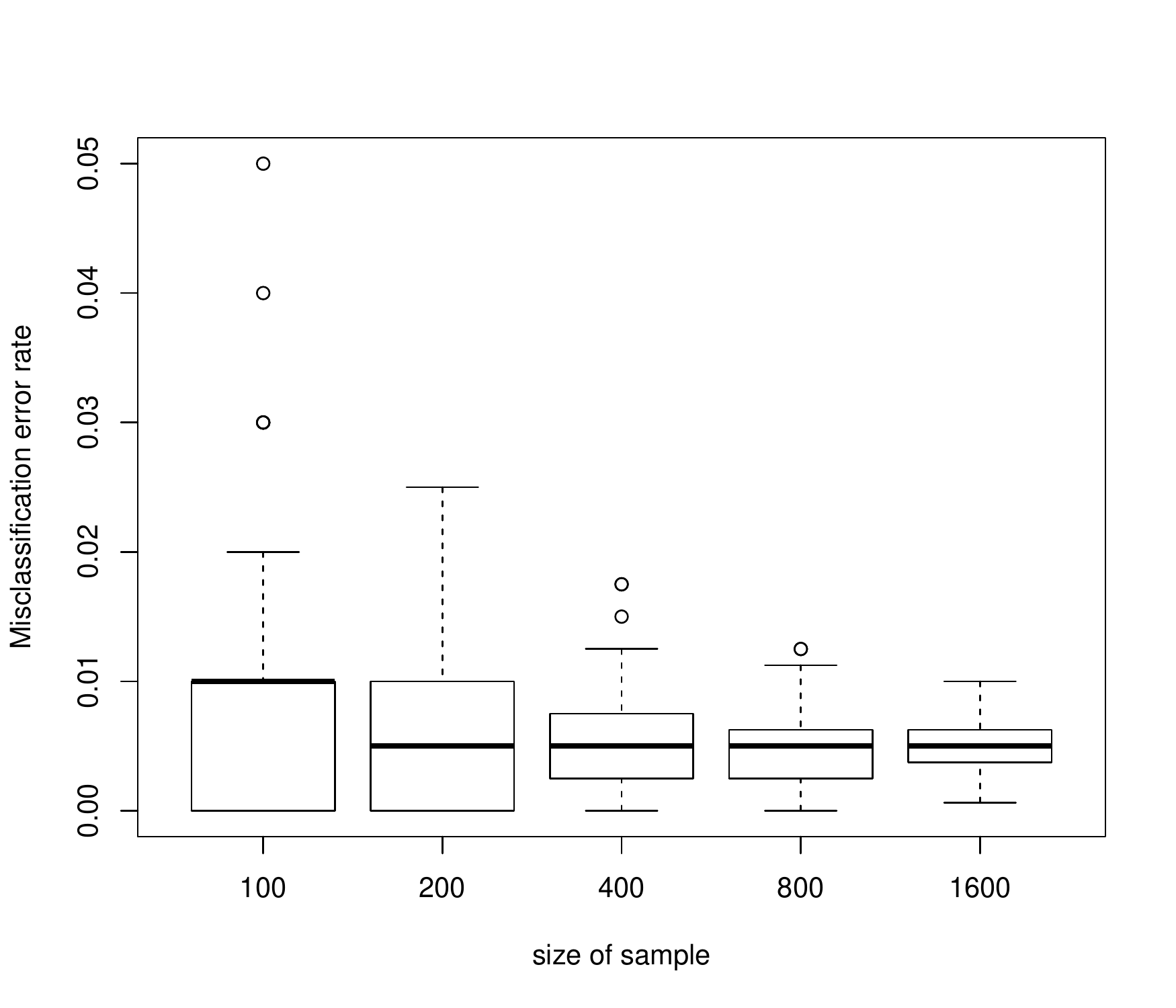} \label{simul3error}                }
  \caption{Boxplots of indicators of the good behavior of the estimation procedure for different sample sizes.\label{simul3}}
\end{figure}

\subsection{Robustness}
Samples are generated from the bivariate Poisson mixture model  \citep{Kar08} with $\bpi=(1/3,2/3)$, whose one-dimensional marginal parameters $\bal_k=(\lambda_{k1},\lambda_{k2},\lambda_{k3})$ take on the following values: $\lambda_{1h}=h$  and  $\lambda_{2h}=3+h$, for  $h=1,2,3$ (see notation detailed in \citet{Kar08}). Figure~\ref{karlis} presents the results and shows the robustness of the Gaussian copula mixture model since it efficiently manages such data sets, as detailed now.

As shown by Figure~\ref{karlis_err}, the resulting misclassification error rate converges to the theoretical misclassification error rate (equal to $0.0967$). Moreover, Figure~\ref{karlis_kl} shows that the Kullback-Leibler divergence almost vanishes when the sample size increases, thus demonstrating the flexibility of the Gaussian copula mixture model. Moreover, the resulting parameters reflect the main properties of the true distribution. Indeed, 
Figure~\ref{karliscorrel} shows that the correlation coefficient between both variables, for component 1, converges to its theoretical value (equal to  $\lambda_{11} + \lambda_{13}=4$). In the same way, Figure~\ref{karlis_alpha} shows that the one-dimensional margin parameter of variable 1, for component 1, converges to its theoretical value (equal to $\lambda_{13}/\sqrt{(\lambda_{11} + \lambda_{13})(\lambda_{12} + \lambda_{13})}=3/\sqrt{20}\simeq 0.67$).

Finally, the estimation procedure is not too much time consuming since samples of size $n=100$ and $n=1600$ require respectively $12$ and $54$ seconds, on an Intel Core i5-3320M processor.

\begin{figure}[!ht]
  \centering
  \subfigure[Misclassification error rate]{\includegraphics[scale=0.4]{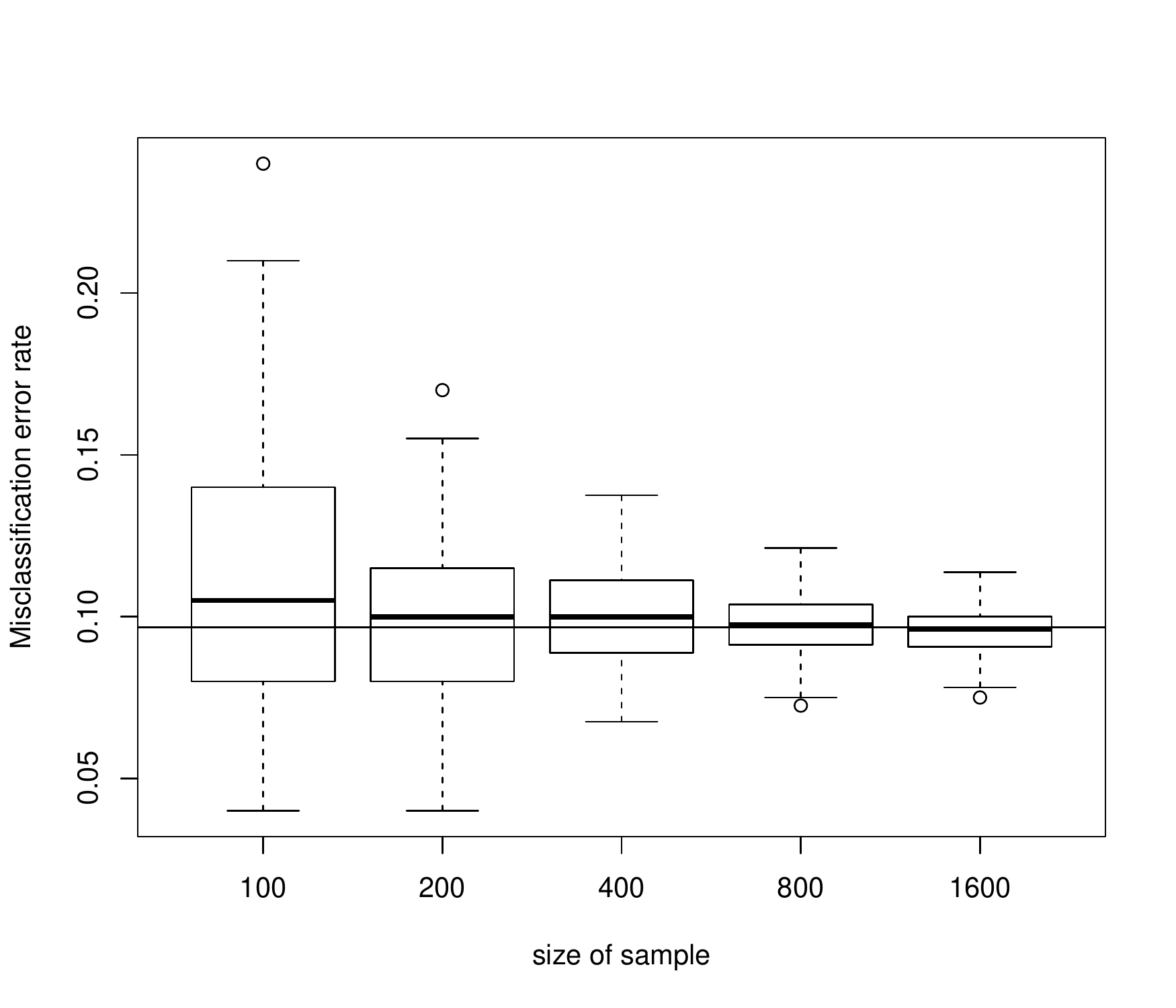}\label{karlis_err} }  
    \subfigure[Kullback-Leibler divergence from the true model]{\includegraphics[scale=0.4]{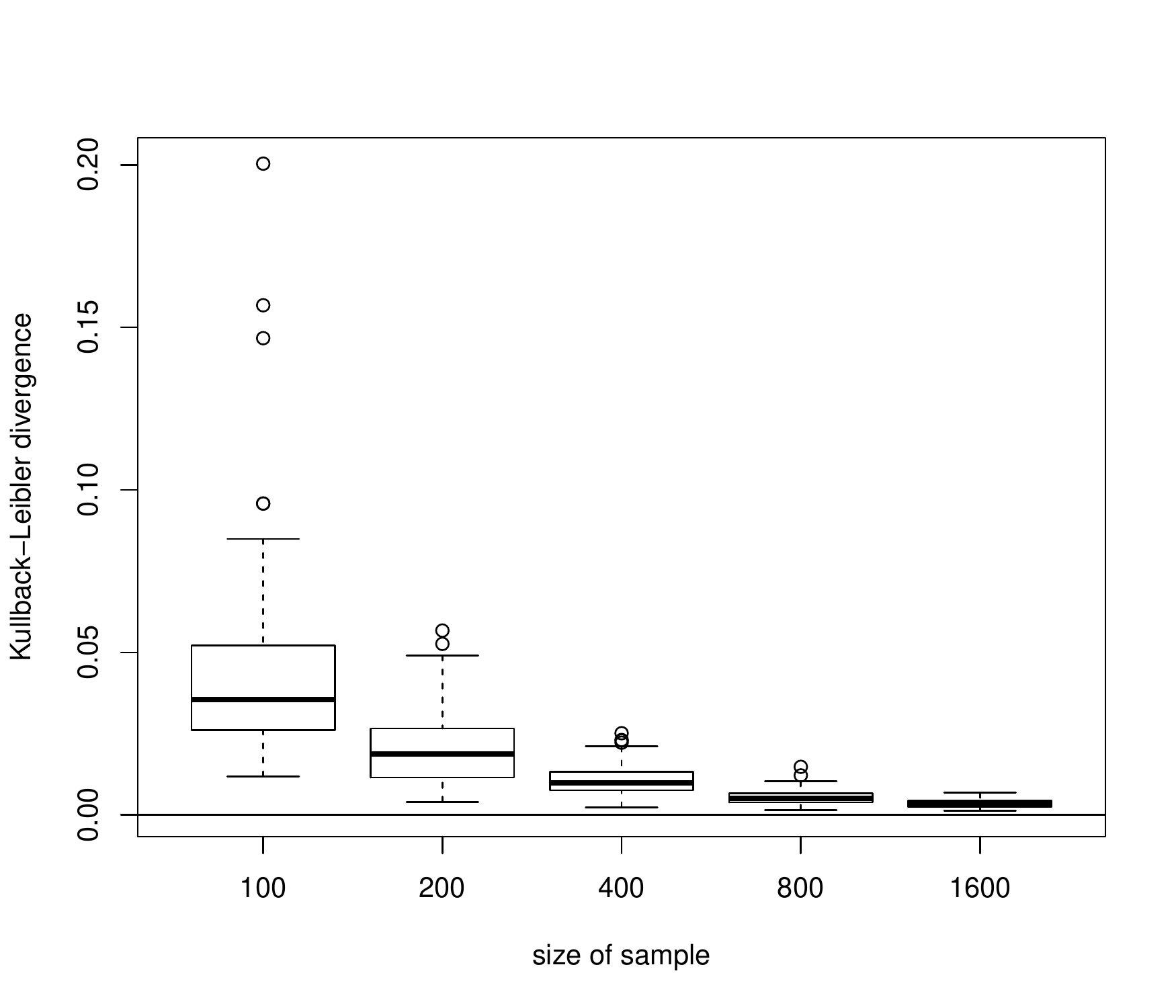} \label{karlis_kl}}              \\
    \subfigure[Correlation coefficient between both variables for component 1]{\includegraphics[scale=0.4]{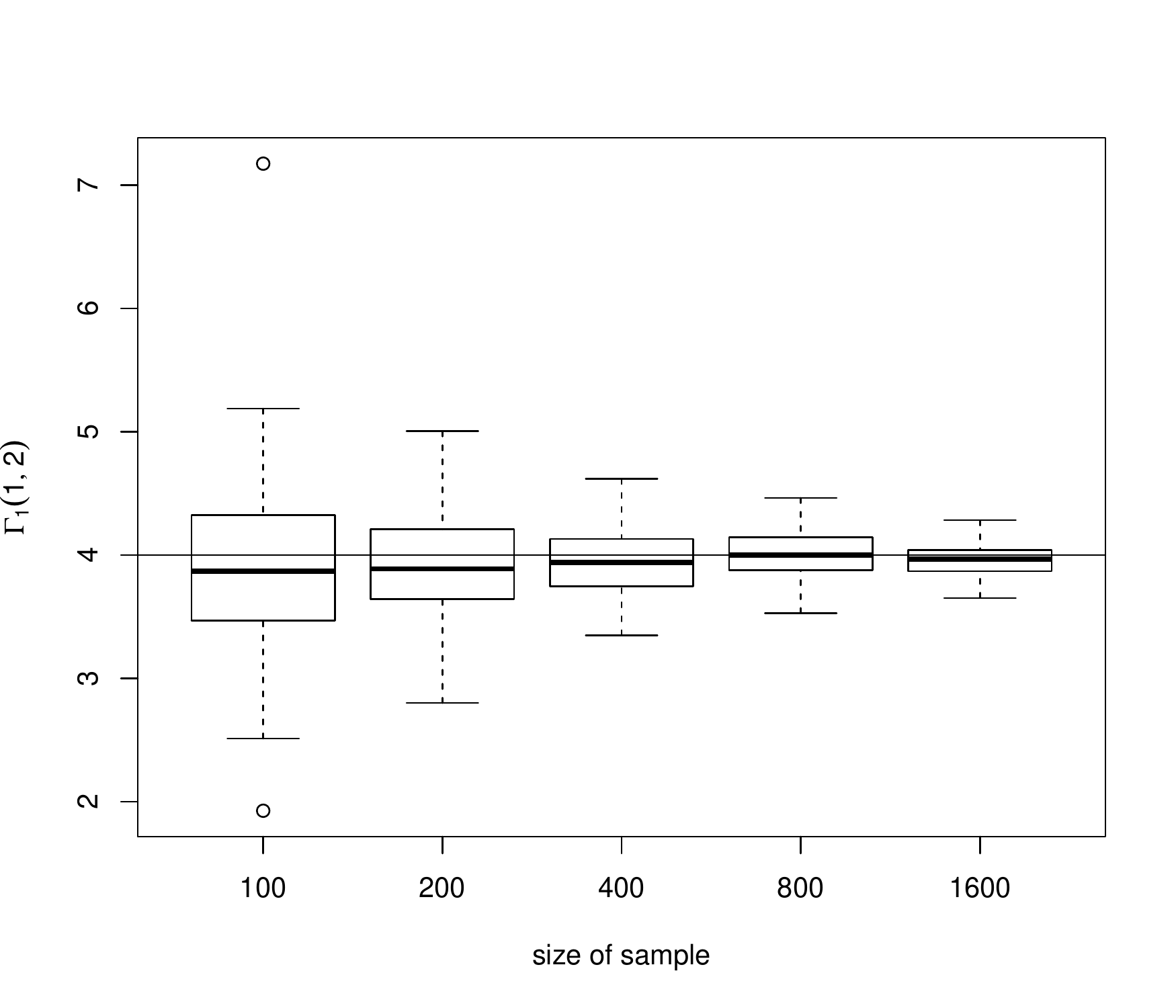} \label{karliscorrel}}
  \subfigure[One-dimensional margin parameter of variable 1 for component 1]{\includegraphics[scale=0.4]{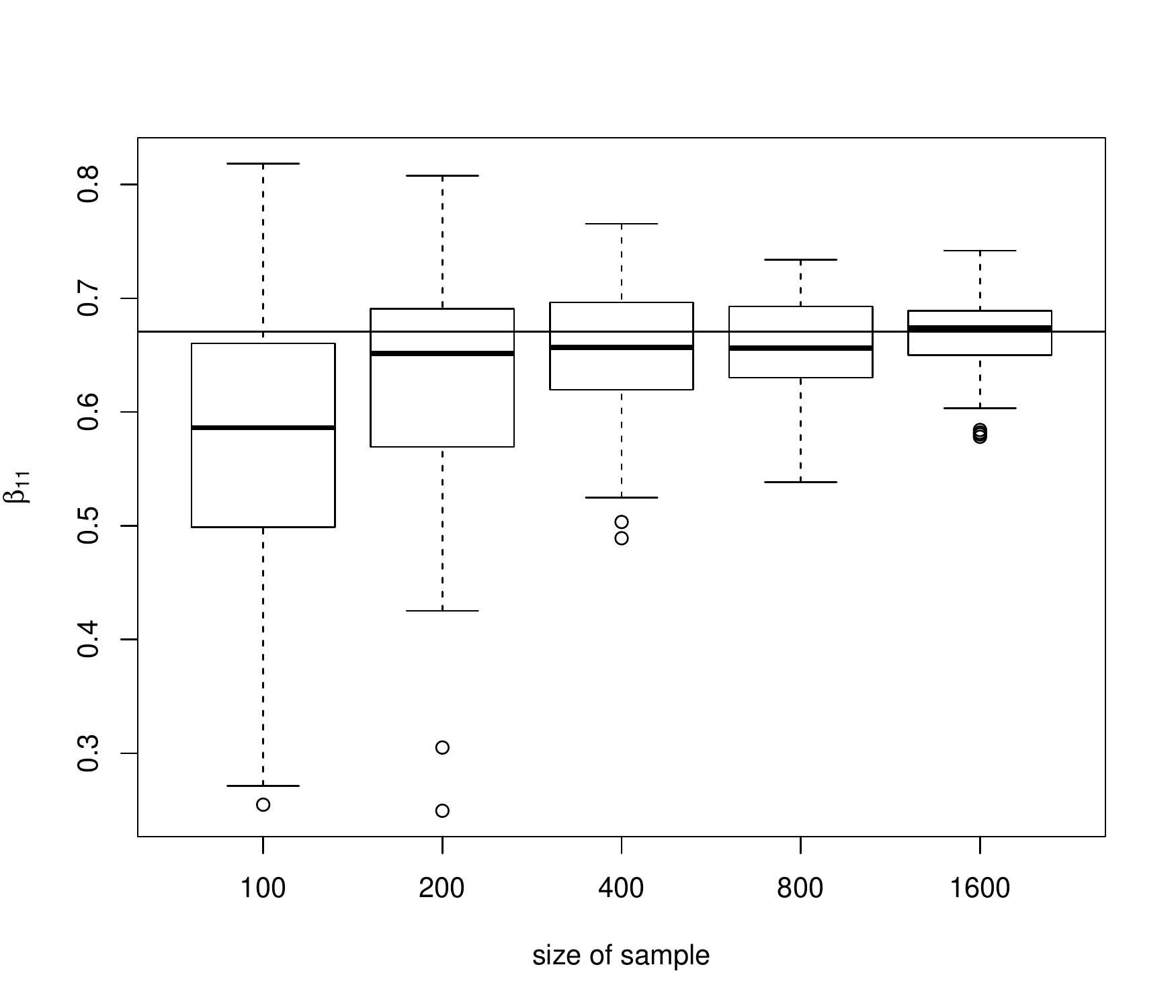} \label{karlis_alpha} }
  \caption{Boxplots of the  indicators related to the estimated model. Values obtained with the true Poisson mixture model are indicated by the horizontal black lines.}\label{karlis}
\end{figure}

\section{Applications} \label{app}
In this section, we analyze two real data sets with the proposed Gaussian copula mixture model. For each number of components, 10 runs of Gibbs sampler are performed under the conditions described in the previous section. 

Model selection is performed by using two information criteria (BIC criterion \citep{Sch78}, ICL criterion \citep{Bie00}).
These asymptotic criteria are computed with the estimate provided by the Gibbs sampler and not with the maximum likelihood estimate. However, this procedure remains valid since the BIC criterion can be computed with the estimate of the maximum \emph{a posteriori} \citep{Leb06}. Moreover, the ICL criterion can be computed by penalizing the BIC criterion with a term of entropy.

These criteria require the computation of the number of parameters. The locally independent model involves $\nu_{\text{Loc}}=(g-1)+g\sum_{j=1}^e\nu_j$ parameters where $\nu_j$ is the number of parameters involved by the one-dimensional marginal distribution of one component (\emph{i.e.} $\nu_j=2$ if $x^j$ is continuous, $\nu_j=1$ if $x^j$ is integer and $\nu_j=m_j-1$ if $x^j$ is ordinal with $m_j$ modalities). The heteroscedastic Gaussian copula mixture model involves $\nu_{\text{He}}=\nu_{\text{Loc}} + g\frac{e(e-1)}{2}$ parameters and the homoscedastic Gaussian copula mixture model requires $\nu_{\text{Ho}}=\nu_{\text{Loc}} + \frac{e(e-1)}{2}$ parameters.

\subsection{South African Hearth data set}
\subsubsection*{The data} 
Data are available at \url{http://sci2s.ugr.es/keel/dataset.php?cod=184}.
This data set is a retrospective sample of males in a heart-disease high-risk region of the Western Cape, South Africa.  Many of the coronary heart disease  (CHD) positive men have undergone blood pressure reduction treatment and other programs to reduce their risk factors after their CHD event. In some cases the measurements were made after these treatments. The class label indicates if the person has a coronary heart disease (negative or positive) and is hidden for our analysis. 
Individuals are described by the following nine variables. The continuous variables are systolic blood pressure (sbp), cumulative tobacco (tobacco), low density lipoprotein cholesterol (ldl), adiposity, obesity and current alcohol consumption (alcohol). The integer variables are type-A behavior (typea) and  age at onset (age). Finally, the binary variable indicates the presence or not of heart disease in the family history (famhist).

\subsubsection*{Model selection} 
Three mixture models (locally independent, heteroscedastic and homoscedastic mixture of Gaussian copulas) are fitted for various numbers of components. Table~\ref{wine_bic} presents the values of information criteria used to select the homoscedastic tri-component Gaussian copula mixture model.
Indeed, this model obtains the best results for fitting the data distribution (BIC) and for fitting the best partition (ICL). Note that the locally independent model with five components obtains a value of the BIC criterion close to the tri-component homoscedastic model, but the ICL criterion favors this latter. Moreover, this model detects less components than the locally independent model, thus its interpretation is easier.

\begin{table}[!ht]
\centering
\begin{tabular}{cccccccc}
 & g & 1 & 2 & 3 & 4 & 5  \\ 
\hline BIC & loc. indpt.  & -14127.26  & -13131.88  & -12813.92 & -12829.68  &\textbf{-12738.66}  \\
 & homo. & -14724.98 & -13016.09  & \textbf{-12739.94}  & -12774.15  & -12927.45  \\
 & hetero.  & -14724.98  & -13076.93  & \textbf{-12971.72}  & -13071.92 & -13253.06  \\
\hline ICL & loc. indpt.  & -14127.26  & -13144.21 & -12832.12  & -12887.19  & \textbf{-12805.68} \\
 & homo. &  -14724.98 & -13028.07 & \textbf{-12762.79}  & -12816.44  &  -12979.06  \\
 & hetero.  & -14724.98  & -13085.52  & \textbf{-12989.06}  & -13103.61  & -13299.16 \\
\hline 
\end{tabular} 
\caption{Values of the BIC and ICL criteria obtained on the South African Hearth data set (best values are in bold).\label{wine_bic}}
\end{table}

\subsubsection*{Partition study}
The three models overestimate the number of components (true number is two). However, the Gaussian copula mixture model finds number of components more relevant than the locally independent model. If we assume that the number of component is known, the misclassification error rates is slightly better for both copulas models (homo. and hetero.) in comparison to the locally independence model (see Table~\ref{tab:saconfu}). Moreover, the equality constraint of the covariance matrix increases the value of the BIC criterion obtained by the homoscedastic model, but it also impacts its resulting partition according to the partition resulting from the heteroscedastic model (see Table~\ref{tab:saconf2}).

\begin{table}[!ht]
\begin{center}
\begin{tabular}{ccccccc}
CHD & \multicolumn{2}{c}{loc. indpt.} & \multicolumn{2}{c}{homo.} & \multicolumn{2}{c}{hetero.} \\ 
 & Class 1 & Class 2 & Class 1 & Class 2 & Class 1 & Class 2 \\ 
 \hline negative & 187 & 149 & 191 & 150 & 195 & 148 \\ 
positive & 115 & 11 & 111 & 10 & 107 & 12 \\ 
\hline
\end{tabular} 
\end{center}
\caption{Confusion matrices between the CHD statu and the partitions provides by the three bi-component mixture models.}\label{tab:saconfu}
\end{table}

\begin{table}[!ht]
\begin{center}
\begin{tabular}{ccccc}
&& \multicolumn{3}{c}{hetero.} \\ 
& & Class 1 & Class 2 & Class 3  \\ 
\hline 
& Class 1 & 46 & 11 & 0 \\ 
homo.& Class 2 & 0 & 92 & 14 \\ 
 &Class 3 & 0 & 3 & 296 \\ 
\hline 
\end{tabular} 
\end{center}
\caption{Confusion matrices between the tri-components homoscedatic model (row) and the tri-component homoscedastic model (column).}\label{tab:saconf2}
\end{table}

Figure~\ref{saheartACP} shows the PCA visualization for the component 3.
The tri-component homoscedastic Gaussian copula mixture model provides three well-separated classes as shown by the factorial representation of Figure~\ref{saind}. We can see that the class 2 (red triangles) is an "intermediate" class, while class 1 (red dots) is strongly different to class 3 (black dots). We now detail the model interpretation.

\begin{figure}[!ht]
  \centering
  \subfigure[Scatterplot of the individuals]{\includegraphics[scale=0.4]{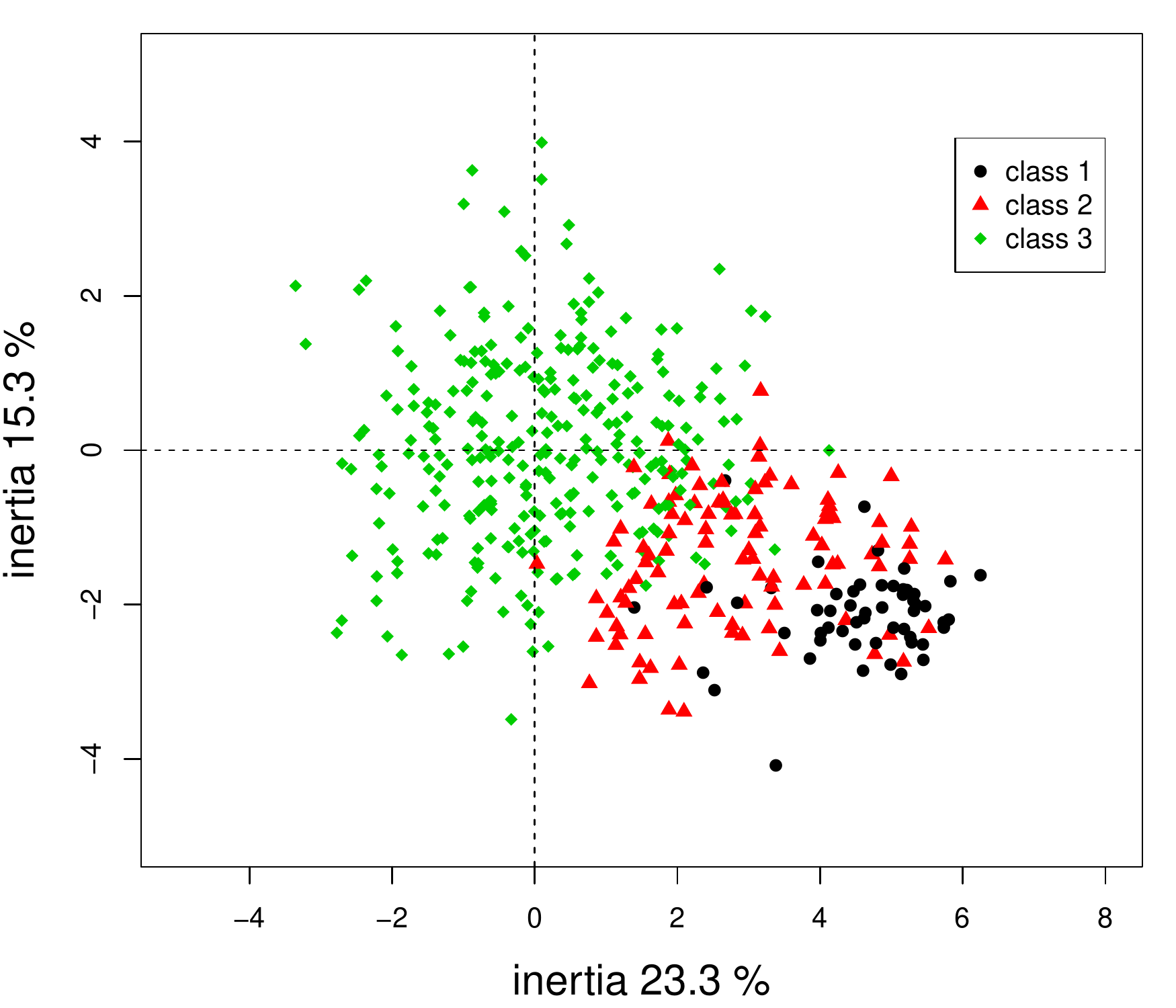} \label{saind}}
  \subfigure[Correlation circle]{\includegraphics[scale=0.4]{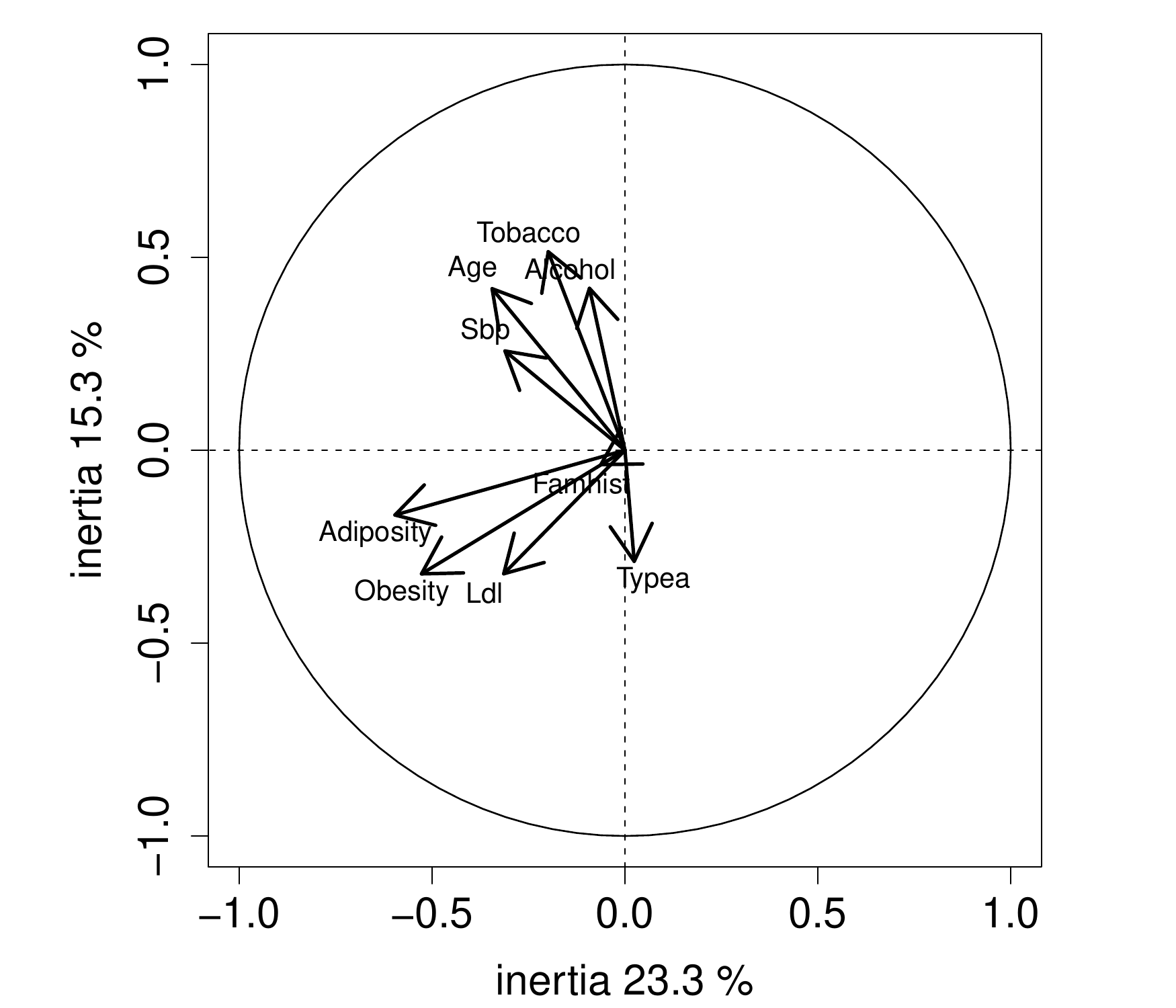} \label{sevar}}
  \caption{Visualization throughout parameters of component 3 for the tri-component homoscedastic Gaussian copula mixture model.}\label{saheartACP}
\end{figure}
 
\subsubsection*{Interpretation of best-fit model}
The PCA visualization (Figure~\ref{saheartACP}) shows that individuals of class 3 have a riskier behaviour than the others (large consumption of tobacco, alcohol, older population, large level of obesity).

 A three-level interpretation (proportions, one-dimensional marginal distributions and intra-class dependencies) is feasible using the parameters. The main characteristic variables are summarized in Figure~\ref{wine_corr}:
 \begin{itemize}
 \item Class~1 (\emph{weak-risk behaviours}): this is the minority class ($\pi_1=0.07$). It groups the young individuals having a short consumption alcohol and tobacco. This class groups 57 individuals where only one has a coronary heart disease.
 \item Class~2 (\emph{moderate-risk behaviours}): this class have a moderated proportion ($\pi_2=0.24$). This class is composed with individuals having a strong alcohol consumption. The other variables take intermediate values. Among the 106 individuals belonging to this class, 20 have a coronary heart disease.
 \item Class~3 (\emph{high-risk behaviours}): this is the majority class ($\pi_3=0.69$). It groups the individuals having the more dangerous behaviour. Among the 299 individuals belonging to this class, 139 have a coronary heart disease.
 \end{itemize}
 
Finally, for all the components, age and the consumptions of alcohol and tobacco are strongly linked (see Figure~\ref{sevar}. Moreover, adiposity and obesity are also strongly linked.

\begin{figure}[!ht]
  \centering \includegraphics[scale=0.6]{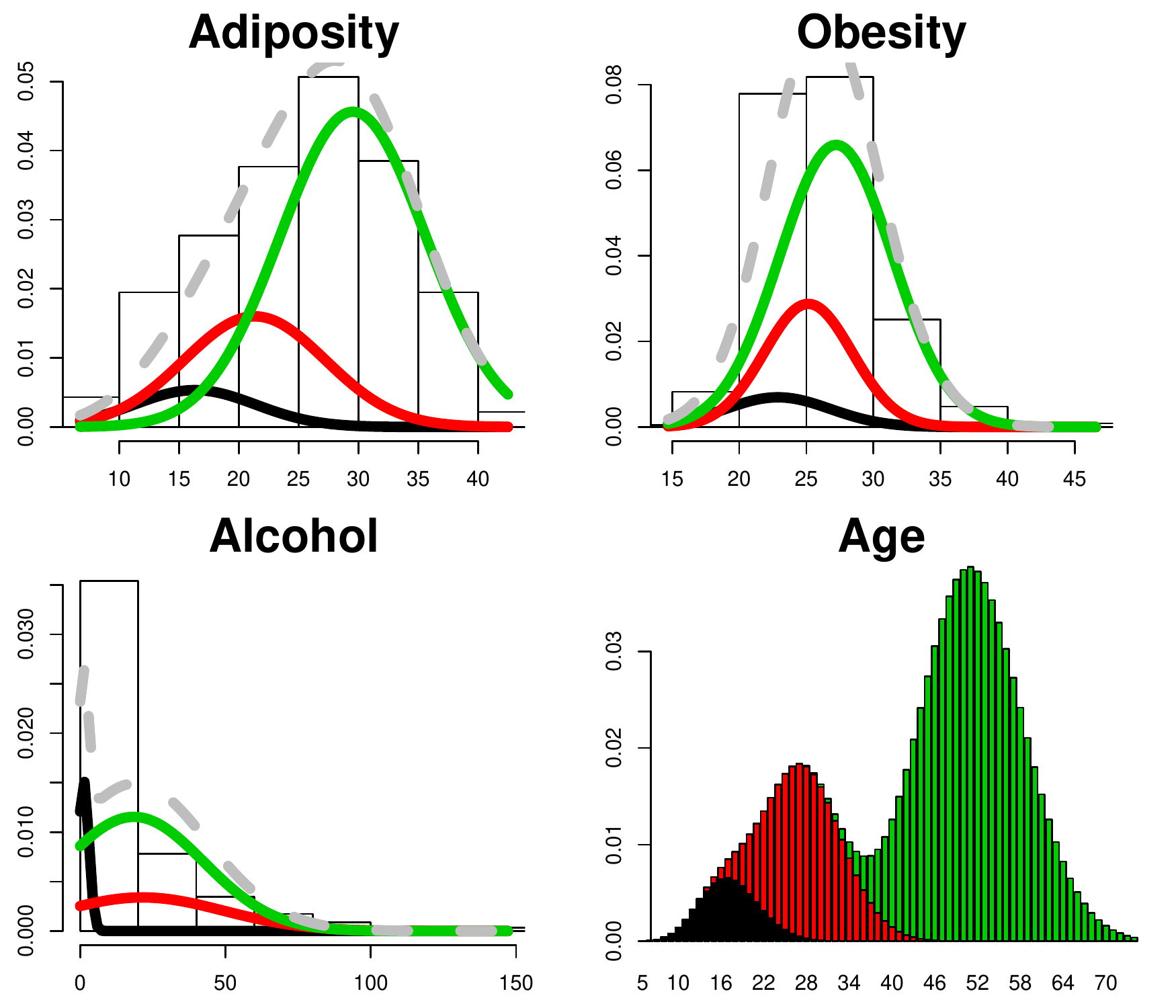}
  \caption{One dimensional marginal distributions for the South Africa Heart data: whole distribution (gray), component~1 (black), component~2 (red) and component~3 (green) of the tri-component homoscedastic Gaussian copula mixture model.}\label{wine_corr}
\end{figure}

\subsubsection*{Conclusion}
For such data, the Gaussian copula mixture model reduces the drawbacks of the locally independent model. By decreasing the number of components, it yields a more interpretable model that better fits the data (BIC criterion) and provides a pertinent partition (ICL criterion). 
Finally, the estimation of main intra-class dependencies, based on PCA outputs per component, is an efficient tool for refining the interpretation.

\subsection{Forest fire data set}
\subsubsection*{The data} 
Data are composed of 517 forest fires that have occurred in the northeast region of Portugal \citep{Cor07}. These forest fires are described by the following meteorological variables: seven continuous variables (four fire weather index (FWI) system variables, \emph{i.e.} fine fuel moisture code (FFMC), duff moisture code (DMC), drought code (DC), initial season index (ISI), and three meteorological variables, \emph{i.e.} temperature (Temp), relative humidity (RH) and wind) and three binary variables indicating the presence of rain, the season (summer or other) and the day of the week (weekend or other).

\subsubsection*{Model selection} 
Table~\ref{fire_bic} presents the values of information criteria used to distinctly select the heteroscedastic tri-component Gaussian copula mixture model.
Note that the locally independent model degenerates with five components.
\begin{table}[htp]
\centering
\begin{tabular}{ccccccccc}
 & g & 1 & 2 & 3 & 4 & 5    \\ 
\hline BIC & loc. indpt.  & -15152.95 & -14164.51 & \textbf{-13990.27 }& -14068.92 & NA\\
 & homo.  & -14401.80 &\textbf{ -13751.82} & -13927.05 & -13986.90 & -13755.69\\
 & hetero &  -14401.80 & -13781.86 & \textbf{-13680.67 }& -13846.63 & -13745.84\\
\hline ICL & loc. indpt.  & -15152.95 & -14170.97 & \textbf{-14022.49} & -14131.22 & NA\\
 & homo.  & -14401.80 & \textbf{-13756.76} & -13956.68 & -14070.49 & -13774.41\\
 & hetero &  -14401.80 & -13785.33 &\textbf{ -13682.68 }& -13885.11 & -13776.81\\
\hline 
\end{tabular} 
\caption{Values of the BIC and ICL criteria obtained on the forest fire data set (best values are in bold).\label{fire_bic}}
\end{table}

The heteroscedastic model with three components obtains better values of the information criteria than the locally independent model since it models the intra-component dependencies. Moreover, as shown by Table~\ref{fire_confu}, these intra-component dependencies influence the resulting partition since four individuals are affiliated in different classes by both models.

\begin{table}[ht!]
\begin{center}
\begin{tabular}{rccc}
& \multicolumn{3}{c}{hetero.} \\		
 &  Class 1 & Class 2 & Class 3 \\ 
\hline Class 1 & 33 & 1 & 0 \\ 
loc. indpt. Class 2 & 1 & 402 & 0 \\ 
 Class 3 & 2 & 0 & 78 \\ 
 \hline
\end{tabular} 
\end{center}
\caption{Confusion matrix between the partition provides by the locally independent model with three components (rows) and the heteroscedastic Gaussian copula mixture with three components (columns).}\label{fire_confu}
\end{table}

\subsubsection*{Interpretation of best-fit model}
The three-step interpretation of the  heteroscedastic tri-component Gaussian copula mixture model is presented by using the parameters summarized in Figure~\ref{fire_corr}:
\begin{itemize}
\item Class~1 (\emph{unpredictable fires}): this is the minority class ($\pi_1=0.09$). It groups fires difficulty predictable since they occur with small values of the four FWI system variables (especially ISI). These fires occur during all the year only with dry weather.
\item Class~2 (\emph{predictable fires of summer}): this is the majority class ($\pi_2=0.78$). This class is composed with fires occurring mainly in summer. They appear with high values of the four FWI system variables and a high temperature.
\item Class~3 (\emph{fires of winter}): this is a moderate class ($\pi_3=0.13$). It groups fires occurring in winter, so with small temperature. They occur with a dry weather and very small values of the four FWI systems (except ISI).
\end{itemize}

The correlation matrices highlight dependencies between the summer period and high temperatures, and between FFMC and DMC values (see Figure~\ref{firevar} for component 3). Finally, it should be noted that the variable indicating the day of the week roughly follows the same distribution for all three classes.

\begin{figure}[!ht]
  \centering \includegraphics[scale=0.6]{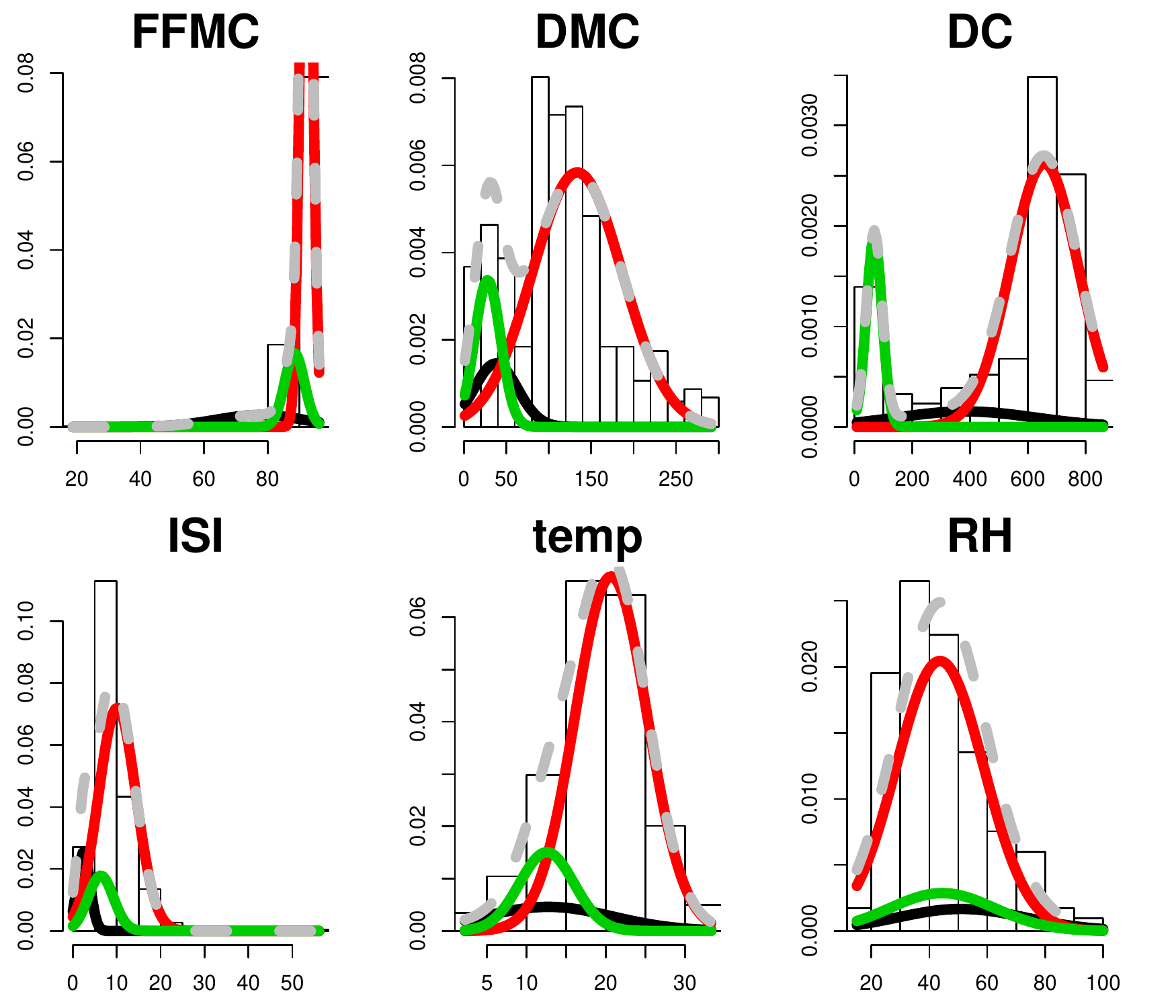}
  \caption{One dimensional marginal distributions for the forest fire data: whole distribution (gray), component~1 (black), component~2 (red) and component~3 (green) of the tri-component homoscedastic Gaussian copula mixture model.}\label{fire_corr}
\end{figure}

The results of the PCA done according to component 3 is shown in Figure~\ref{fire_pca}. Obviously, Figure~\ref{fire_ind} shows that the individuals belonging to component~3 are strongly different to the other ones. Indeed, few individuals belonging to component~2 are visible in this map. The other individuals are too far away from the origin. Thus, the distribution of component~3 is strongly different from the other distributions.
\begin{figure}[!ht]
  \centering
  \subfigure[Scatterplot of the individuals]{\includegraphics[scale=0.4]{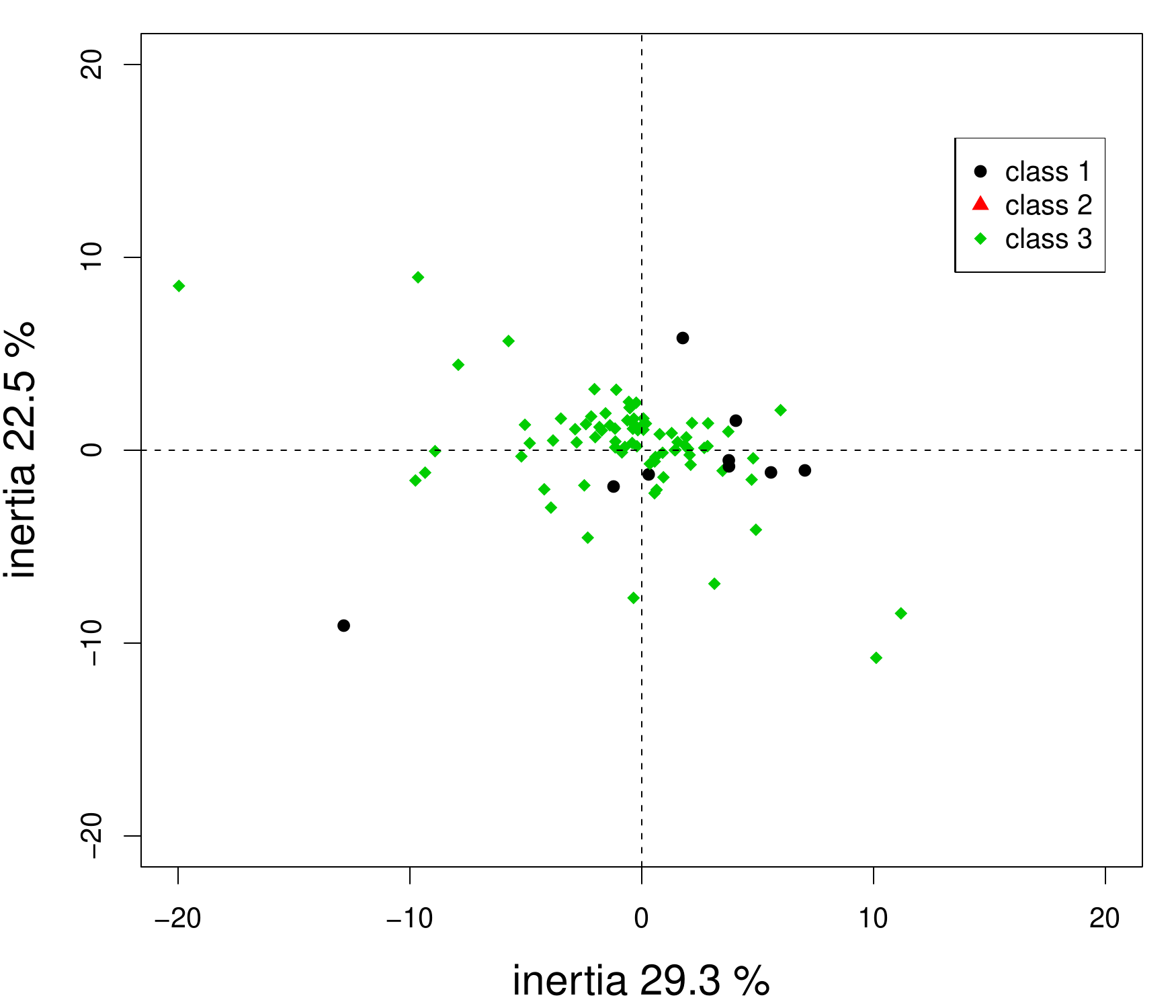} \label{fire_ind}}
  \subfigure[Correlation circle]{\includegraphics[scale=0.4]{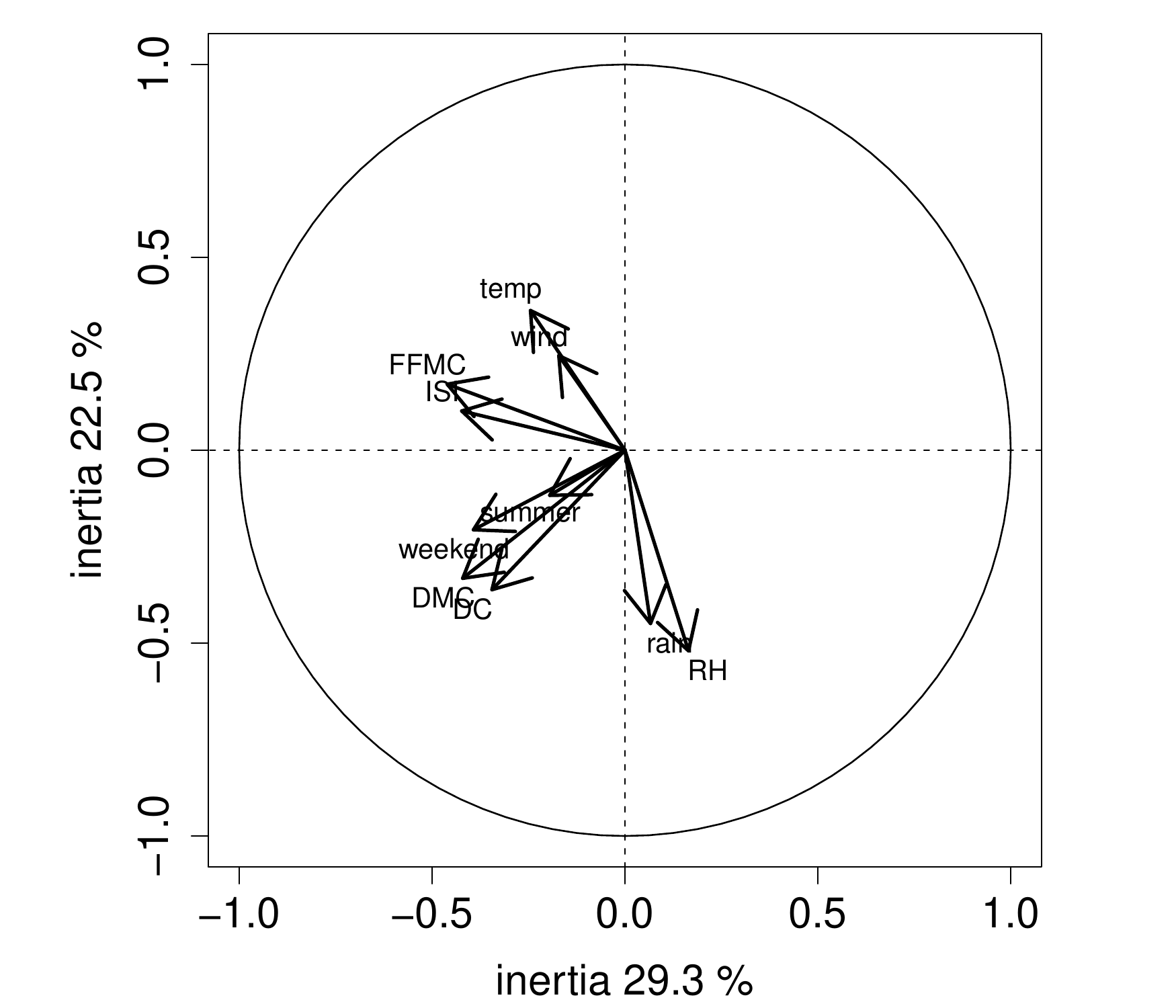} \label{firevar}}
  \caption{Visualization throughout parameters of component 3 for the tri-component heteroscedastic Gaussian copula mixture model.}\label{fire_pca}
\end{figure}

\subsubsection*{Conclusion}
The cluster analysis obtained with the Gaussian copula mixture model is more accurate than that obtained with the locally independent model. Indeed, it provides a meaningful model which needs less components. This model puts the light on three kinds of fires: the fires predictable with the  FWI system (class~2), the fires occurring in winter (class~3) and the unpredictable fires (class~1).

\section{Conclusion and future extensions} \label{conclusion}
A Gaussian copula mixture model has been introduced and used to cluster mixed data. Using Gaussian copulas, the one-dimensional marginal distributions of each component follow conventional distributions, and intra-class dependencies are effectively modelled. Thus the model results can be easily interpreted in three steps, as in the case of models developed for data sets composed of a single type of variable. Using the continuous latent variables of Gaussian copulas, a PCA-type method displays for component-based visualization of individuals. Moreover, this approach provides a summary of intra-component dependencies, which avoids fastidious interpretation of correlation matrices. 

In the description of numerical experiments and applications, we pointed out that this model is sufficiently robust to fit data obtained from another model. Furthermore, it can reduce the bias produced by the locally independent model (\emph{e.g.} reduction of the number of component).

The number of parameters increases with the number of components and number of variables, particularly due to the correlation matrices of the Gaussian copulas. In order to overcome this drawback, we have proposed an homoscedastic version of the model assuming equality between correlation matrices.  However, the number of parameters required by this model can stay large when the number of variables increases. Therefore, more parsimonious correlation matrices could be proposed to overcome this drawback in future studies.

Since the distribution of all the variables is modeled, this model could be used to manage data sets with missing values. By assuming that values are missing at random,  the Gibbs sampler could also be adapted, but the underlying principle would remain roughly the same.

Finally, the proposed model cannot cluster non-ordinal categorical variables having more than two modalities. In such cases, the cumulative distribution function is not defined. An artificial order between modalities could be added to define a cumulative distribution function, but this method presents three potential difficulties that require attention: it assumes regular dependencies between the modalities of two variables, its estimation would slow down the estimation algorithm, and its stability would have to be verified.

\textbf{MixCluster} (\url{https://r-forge.r-project.org/R/?group_id=1939}) is an R package which performs the cluster analysis method described in the article. It also contains the data sets used  in this paper.

\bibliography{biblio}
\bibliographystyle{elsarticle-harv}

\appendix
\section{Proof of the model identifiability}\label{identifiability}
The model identifiability is proved by two propositions. The first proposition proves the model identifiability when the variables are continuous and/or integer. This proposition presents the reasoning in a simple case since it does not consider the ordinal variables. The second proposition proves that the model requires at least one continuous or integer variable to be identifiable.

\begin{prop}[Identifiability with continuous and integer variables]
The  Gaussian copula mixture model is weakly identifiable \citep{Tei63} if the variables are continuous and integer ones (\emph{i.e.} the one-dimensional marginal distributions of the components are Gaussian or Poisson distributions). Thus, 
\begin{align}
\forall \bx \in \mathbb{R}^c \times \mathbb{N}^d, \quad & \sum_{k=1}^g \pi_k p(\bx|\bal_k) = \sum_{k=1}^{g'} \pi_k'p(\bx |\bal_k') \label{ident_assumption}\\
\Rightarrow  \quad& g=g',\; \bpi=\bpi',\; \bal=\bal'.
\end{align}
\end{prop}

\begin{proof}
The identifiability of the multivariate Gaussian mixture models and of the univariate Poisson mixture model \citep{Tei63,Yak68} involves that \eqref{ident_assumption} implies
\begin{equation}
g=g',\; \bpi=\bpi',\; \bbet_{kj}=\bbet_{kj}'\text{ and } \bG_{k\textsc{c}\textsc{c}}=\bG_{k\textsc{c}\textsc{c}}'.
\end{equation}
We now show that $\bG_{k\textsc{c}\textsc{d}}=\bG_{k\textsc{c}\textsc{d}}'$ and $\bG_{k\textsc{d}\textsc{d}}=\bG_{k\textsc{d}\textsc{d}}'$.

Let $j\in\{1,\ldots,c\}$ and $h\in\{c+1,\ldots,e\}$. We denote by $\rho_k=\bG_k(j,h)$, $\rho_k'=\bG_k'(j,h)$,  $v_k=\Phi_1^{-1}(P(x^j;\bbet_{kj}))$, $\varepsilon_k(x^j)=\pi_k\frac{\phi_1(v_k)}{\sigma_{kj}}$,  $a_k=\frac{b_k^{\oplus}(x^j)-\rho_k v_k}{\sqrt{1-\rho_k^2}}$ and $a_k'=\frac{b_k^{\oplus}(x^j)-\rho_k' v_k}{\sqrt{1-\rho_k'^2}}$. Without loss of generality, we order the components as such $\sigma_{kj}>\sigma_{k+1j}$ and if $\sigma_{kj}=\sigma_{k+1j}$ then $\mu_{kj}>\mu_{k+1j}$, then \eqref{ident_assumption} implies that
\begin{equation*}
1+\sum_{k=2}^g (\varepsilon_k (x^j)\Phi(a_k))/(\varepsilon_1 (x^j)\Phi(a_1))=\sum_{k=1}^g \varepsilon_k(x^j) \Phi(a_k')/(\varepsilon_1 (x^j)\Phi(a_1)).
\end{equation*}
Let $\gamma_t= \{ (x^j,x^h) \in \mathbb{R} \times \mathbb{N}: \; a_1=t\}$. Then, letting $x^h\rightarrow \infty$ as such $(x^j,x^h)\in \gamma_t$,
\begin{equation}
\forall t, \quad \frac{\int_{t}^{a_1'}\phi(u)du}{\Phi(t)}=0.
\end{equation}
Thus $a_1'=a_1$, so $\rho_{1}'=\rho_{1}$. Repeating this argument for $k=2,\ldots,g$ and for all the couples $(j,h)$, we conclude that $\bG_{k\textsc{c}\textsc{d}}=\bG_{k\textsc{c}\textsc{d}}'$.

When both variables are integer, we use the same argument with $\gamma_{(t,\xi)}= \{ (x^j,x^h) \in \mathbb{N} \times \mathbb{N}: \; a_1\in B(t,\xi) \}$. Note that if $\rho_{1}\neq \rho_1'$ then $\exists n_0$ as such $\forall x^j>n_0$ $a_1'>t+\xi$. Letting $x^h\rightarrow \infty$ as such $(x^j,x^h)\in \gamma_{(t,\xi)}$, we obtain the following contradiction
$
 \frac{\int_{t+\xi}^{a_1'}\phi(u)du}{\Phi(t-\xi)}=0 \text{ and } \frac{\int_{t+\xi}^{a_1'}\phi(u)du}{\Phi(t-\xi)}>0.
$
So, $a_1'=a_1$ then $\rho_1=\rho_1'$. Repeating this argument for $k=2,\ldots,g$ and for all the couples $(j,h)$, we conclude that $\bG_{k\textsc{d}\textsc{d}}=\bG_{k\textsc{d}\textsc{d}}'$. 
\end{proof}

\begin{prop}[Identifiability of the  Gaussian copula mixture model] \label{propindent}
The  Gaussian copula mixture model is weakly identifiable \citep{Tei63} if at least one variable is continuous or integer.
\end{prop}

\begin{proof}
In this proof, we consider only one continuous variable and two binary variables. Obviously, the same reasoning can be extend to the other cases.
We now show that $\bG_{k\textsc{c}\textsc{d}}=\bG_{k\textsc{c}\textsc{d}}'$ and $\bG_{k\textsc{d}\textsc{d}}=\bG_{k\textsc{d}\textsc{d}}'$.

Let $j=1$ and let $h\in\{2,3\}$. We note $\rho_k=\bG_k(j,h)$, $\rho_k'=\bG_k'(j,h)$, $v_k=\Phi_1^{-1}(P(x^j;\bbet_{kj}))$, $\varepsilon_k(x^j)=\pi_k\frac{\phi(v_k;0,1)}{\sigma_{kj}}$,  $a_k=\frac{b_k^{\oplus}(x^j)-\rho_k v_k}{\sqrt{1-\rho_k^2}}$ and 
$a_k'=\frac{b_k'^{\oplus}(x^j)-\rho_k' v_k}{\sqrt{1-\rho_k'^2}}$.
 Without loss of generality, we order the components as such $\sigma_{kj}>\sigma_{[k+1]j}$ and if $\sigma_{kj}=\sigma_{[k+1]j}$ then $\mu_{kj}>\mu_{[k+1]j}$. 
 Note that \eqref{ident_assumption} implies that
$$1+\sum_{k=2}^g (\varepsilon_k (x^j)\Phi(a_k))/(\varepsilon_1 (x^j)\Phi(a_1))=\sum_{k=1}^g \varepsilon_k(x^j) \Phi(a_k')/(\varepsilon_1 (x^j)\Phi(a_1)).$$ Letting $x^1\rightarrow \infty$  and assuming that $\rho_k>0$
then $\frac{\Phi(a_k')}{\Phi(a_k)}=1$. So, $\text{sign}(\rho_k)=\text{sign}(\rho_k').$
 By denoting $\kappa=\lim\limits_{a \to \infty} \frac{\phi(a)}{\Phi(a)}$ and letting $x^1\rightarrow \infty$
$\kappa \frac{1}{\kappa} \frac{\phi(a_k')}{\phi(a_k)}= 1.
$ Thus $a_1'=a_1$, so $\rho_{1}'=\rho_{1}$ and $b_k^{\oplus}(x^j)=b_k'^{\oplus}(x^j)$ so $\bbet_{kh}=\bbet_{kh}'$.

 Note that the same result can be obtain by tending $x^1$ to $- \infty$ is $\rho_k<0$. Repeating this argument for $k=2,\ldots,g$ and for all the couples $(j,h)$, we conclude that $\bG_{k\textsc{c}\textsc{d}}=\bG_{k\textsc{c}\textsc{d}}'$ then $\bG_{k\textsc{d}\textsc{d}}=\bG_{k\textsc{d}\textsc{d}}'$.
\end{proof}

\section{Prior distributions}\label{prior_margin}
We assume independence between the parameters as follows
\begin{equation}
p(\bt)=p(\bpi)\prod_{k=1}^g \left(p(\bG_k) \prod_{j=1}^e p(\bbet_{kj})\right).
\end{equation}

The classical conjugate prior distribution of the proportion vector is the Jeffreys non informative one which is the following Dirichlet distribution 
\begin{equation}
\bpi \sim \mathcal{D}_g\left(\frac{1}{2},\ldots,\frac{1}{2}\right).
\end{equation}

If $x^j$ is \emph{continuous}, then $\bbet_{kj}$ denotes the parameters of a univariate Gaussian distribution so  $p(\bbet_{kj})=p(\mu_{kj}|\sigma_{kj}^2)p(\sigma_{kj}^2)$ with
\begin{equation}
\sigma_{kj}^2\sim \mathcal{G}^{-1}(c_0,C_0)\text{ and }\mu_{kj}|\sigma_{kj}^2\sim \mathcal{N}_1(b_0,\sigma_{kj}^2 /N_0),
\end{equation} 
where $\mathcal{G}^{-1}(.,.)$ denotes the inverse gamma distribution. With an empirical Bayesian approach, the hyper-parameters $(c_0,C_0,b_0,N_0)$ are fixed as proposed by \citet{Raf96}, so $c_0=1.28$, $C_0=0.36\text{Var}(\textbf{x}^j)$, $b_0=\frac{1}{n}\sum_{i=1}^nx_i^j$ and $N_0=\frac{2.6}{\text{argmax } \textbf{x}^j - \text{argmin } \textbf{x}^j}$.

If $x^j$ is \emph{integer}, $\bbet_{kj}$ denotes the parameter of a Poisson distribution and
\begin{equation}
\bbet_{kj}\sim \mathcal{G}(a_0,A_0).
\end{equation}
According to \citet{Fru06}, the values of hyper-parameters $a_0$ and $A_0$ are empirically fixed to $a_0=1$ and $A_0=a_0n/\sum_{i=1}^nx_i^j$.

If $x^j$ is \emph{ordinal}, $\bbet_{kj}$ denotes the parameter of a multinomial distribution  and its Jeffreys non informative conjugate prior involves that \begin{equation}
\bbet_{kj} \sim \mathcal{D}_{m_j}\left(\frac{1}{2},\ldots,\frac{1}{2}\right).
\end{equation}

The conjugate prior of a covariance matrix is the Inverse Wishart distribution denoted by $\mathcal{W}^{-1}(.,.)$. Therefore, it is natural to define the prior of the correlation matrix $\bG_k$ from the prior of the correlation matrix $\boldsymbol{\Lambda}_k$. Indeed, $\bG_k|\boldsymbol{\Lambda}_k$ is deterministic \citep{Hof07}. So,
\begin{equation}
\boldsymbol{\Lambda}_k\sim \mathcal{W}^{-1}(s_0,S_0) \text{ and } \forall 1\leq h,\ell\leq e,\; \bG_k[h,\ell]=\frac{\boldsymbol{\Lambda}_k[h,\ell]}{\sqrt{\boldsymbol{\Lambda}_k[h,h]\boldsymbol{\Lambda}_k[\ell,\ell]}}, \label{prior2}
\end{equation}
where $(s_0,S_0)$ are two hyper-parameters. However, these parameters can not be fitted by an empirical Bayesian approach since $\boldsymbol{y}$ is not observed. Uniform distribution on $]-1,1[$ is also obtained for the margin distributions of each correlation coefficient by setting $s_0=e+1$ and $S_0$ equal to the identity matrix \citep{Bar00}.

\section{Metropolis-within-Gibbs sampler} \label{Metropolis-within-Gibbs sampler}
This section details the four samplings used in Algorithm~\ref{Gibbs}. The first two samplings are difficult to perform directly, so they are done by one iteration of two Metropolis-Hastings algorithms. For both Metropolis-Hastings algorithms, the instrumental distributions assume conditional independence between parameters. So, the smaller are the intra-class dependencies, the closer of the stationary distributions are the instrumental distributions of both  algorithms. Finally, the last two samplings used in Algorithm~\ref{Gibbs} are classical.

\subsection{Class membership and Gaussian vector sampling}
The aim is to sample from \eqref{latent_variables} but the sampling from $\tz,\ty|\tx,\bt^{(r-1)}$ cannot be performed directly. So, it is achieved by one iteration of following Metropolis-Hastings algorithms. The sampling from \eqref{latent_variables} is performed in two steps by using independence between the individuals which involves that
\begin{equation}
p(\tz,\ty|\tx,\bt^{(r-1)})=\prod_{i=1}^n p(z_i|\bx_i,\bt^{(r-1)}) p(\by_i|\bx_i,z_i,\bt^{(r-1)}). \label{ideal1}
\end{equation}
Firstly, each $z_i^{(r)}$ is independently sampled from the multinomial distribution
\begin{equation}
z_i|\bx_i,\bt^{(r-1)} \sim \mathcal{M}_g(t_{i1}(\bt^{(r-1)}),\ldots,t_{ig}(\bt^{(r-1)})),
\end{equation}
where $t_{ik}(\bt^{(r-1)})=\frac{\pi_k^{(r-1)} p(\bx_i|\bal_k^{(r-1)})}{p(\bx_i|\bt^{(r-1)})}$. Note that $t_{ik}(\bt^{(r-1)})$ is the posterior probability that $\bx_i$ arises from component $k$ with the parameters $\bt^{(r-1)}$.

Secondly, each $\by_i^{(r-1/2)}$ is independently sampled given $(\bx_i,\bz_i^{(r)},\bt^{(r-1)})$. Its first $c$ elements, denoted by $\by_i^{\textsc{c}(r-1/2)}$, are deterministically defined by $\by_i^{\textsc{c}(r-1/2)}=\Psi(\bxc_i;\bal_{z_i^{(r)}}^{(r-1)})$. Its  last $d$ elements, denoted by $\by^{\textsc{d}(r-1/2)}_i$, are sampled from the $d$-variate Gaussian distribution $\mathcal{N}_d(\bo,\bG_{z_i^{(r)}}^{(r-1)})$ truncated on the space $\mathcal{S}_{z_i^{(r)}}(\bxd_i)$
\begin{equation}
p(\by^{\textsc{d}}_i|\bx_i,\by_i^{\textsc{c}(r-1/2)},z_i^{(r)},\bt^{(r-1)}) \propto \phi_d(\by^{\textsc{d}}_i;\boldsymbol{\mu}_{z_i^{(r)}}^{\textsc{d}(r-1)},\boldsymbol{\Sigma}_{z_i^{(r)}}^{\textsc{d}(r-1)})\mathds{1}_{\{\by_i^{\textsc{d}} \in \mathcal{S}_{z_i^{(r)}}(\bxd_i)\}},
\end{equation}
where  $\boldsymbol{\mu}_{z_i^{(r)}}^{\textsc{d}(r-1)}=\boldsymbol{\Gamma}_{z_i^{(r)}\textsc{d}\textsc{c}}^{(r-1)}\boldsymbol{\Gamma}_{z_i^{(r)}\textsc{c}\textsc{c}}^{-1(r-1)}\by_i^{\textsc{c}(r-1/2)}$.

The computation of $t_{ik}(\bt^{(r-1)})$ involves the calculation of the integral defined in \eqref{composante} which can be time consuming if $d$ is large ($d>6$). In such a case, the sampling from \eqref{latent_variables} is replaced by one iteration of the Metropolis-Hastings algorithm  which independently samples each couple $(z_i,\by_i)$. Its stationary distribution is
\begin{equation}
p(z_i,\by_i|\bx_i,\bt^{(r-1)})\propto \pi_{z_i}p(\bx_i,\by_i|z_i,\bt^{(r-1)}).
\end{equation}
Note that $p(\bx_i,\by_i|z_i,\bt^{(r-1)})= \phi_e(\by_i;\bo,\bG_{z_i}^{(r-1)})\mathds{1}_{\{\by^{\textsc{c}}_i=\Psi(\bxc_i;\bal_{z_i}^{(r-1)})\}}\mathds{1}_{\{\by_i^{\textsc{d}} \in \mathcal{S}_{z_i}(\bxd_i)\}}$.

The Metropolis-Hastings algorithm samples a candidate $(z^{\star}_i,\by^{\star}_i)$ by the instrumental distribution $q_1(.|\bx_i,\bt^{(r-1)})$ which uniformly samples $z^{\star}_i$ then which samples $\by_i^{\star}|z^{\star}_i$ as follows. Its  first $c$ elements, denoted by $\by_i^{\star \textsc{c}}$, are equal to $\by_i^{\star \textsc{c}}=\Psi(\bxc_i;\bal_{z_i^{\star}}^{(r-1)})$. Its last $d$ elements, denoted by $\by_i^{\star \textsc{d}}$, follow a \emph{multivariate independent Gaussian} distribution truncated on $\mathcal{S}_{z_i^{\star}}(\bxd_i)$. Thus,
\begin{equation}
q_1(z_i,\by_i|\bx_i,\bt^{(r-1)})=\frac{1}{g}\frac{\phi_d(\by_i^{\textsc{d}}|\bo,\boldsymbol{I}) \mathds{1}_{\{\by_i^{\textsc{d}} \in \mathcal{S}_{z_i}(\bxd_i)\}}  }{\prod_{j=c+1}^e p(x_i^j;\bbet_{z_ij}^{(r-1)}) } \mathds{1}_{\{\by_i^{\textsc{c}}=\Psi(\bxc_i|\bal_{z_i}^{(r-1)})\}}.
\end{equation}
The candidate is accepted with the probability
\begin{equation}
\rho_{1i}^{(r)}=\min\left\{ \frac{q_1(z_i^{(r-1)},\by_i^{(r-1)}|\bx_i)}{q_1(z_i^{\star},\by_i^{\star}|\bx_i)} \frac{\pi_{z^{\star}_i} \phi_e(\by_i^{\star};\bo,\bG_{z^{\star}_i}^{(r-1)})}{\pi_{z_i^{(r-1)}} \phi_e(\by_i^{(r-1)};\bo,\bG_{z^{(r-1)}_i}^{(r-1)})} ;1\right\}.
\end{equation}
Thus, at  iteration $(r)$ of Algorithm~\ref{Gibbs}, the sampling according to \eqref{latent_variables} is performed via one iteration of the following Metropolis-Hastings algorithm having $p(z_i,\by_i|\bx_i,\bt^{(r-1)})$ as stationary distribution.
\begin{algo}
\begin{align}
(z^{\star}_i,\by^{\star}_i)& \sim q_1(z,\by|\bx_i)\\
(z^{(r)}_i,\by_i^{(r-1/2)})&=\left\{ \begin{array}{rl}
(z^{\star}_i,\by^{\star}_i) & \text{ with probability } \rho_{1i}^{(r)} \\
(z^{(r-1)}_i,\by_i^{(r-1)})& \text{ with probability } 1-\rho_{1i}^{(r)}.
\end{array}\right.
\end{align}
\end{algo}

\subsection{Margin parameter and Gaussian vector sampling}
The aim is to sample from \eqref{margins} but the sampling from $\bbet_{kj}|\tx,\by^{\bar{\jmath}(r)}_{[rk]},\tz^{(r)},\bbet_{k\bar{\jmath}}^{(r)},\bG_k$ cannot be performed directly. So, it is achieved by one iteration of the following Metropolis-Hastings algorithms. The sampling from \eqref{margins} is performed by using the following decomposition
\begin{multline}
p(\bbet_{kj},\by_{[rk]}^j|\tx,\by^{\bar{\jmath}(r)}_{[rk]},\tz^{(r)},\bbet_{k\bar{\jmath}}^{(r)},\bG_k^{(r-1)})=
p(\bbet_{kj}|\tx,\by^{\bar{\jmath}(r)}_{[rk]},\tz^{(r)},\bbet_{k\bar{\jmath}}^{(r)},\bG_k^{(r-1)})\\\times
p(\by_{[rk]}^j|\tx,\by^{\bar{\jmath}(r)}_{[rk]},\tz^{(r)},\bbet_{k\bar{\jmath}}^{(r)},\bbet_{kj},\bG_k^{(r-1)}).
\end{multline}

Parameter $\bbet_{kj}^{(r)}$ is firstly sampled. The full conditional distribution of $\bbet_{kj}$ is defined up to a normalizing constant such as
\begin{equation}
p(\bbet_{kj}|\tx,\by^{\bar{\jmath}(r)}_{[rk]},\tz^{(r)},\bbet_{k\bar{\jmath}}^{(r)},\bG_k^{(r-1)}) \propto
p(\bbet	_{kj})
\prod_{\{i: z_i^{(r)}=k\}} p(x_i^j|\by_i^{\bar{\jmath}(r)},z_i^{(r)},\bG_k^{(r-1)},\bbet_{kj}). \label{beta_sampling}
\end{equation}
The  distribution of $x_i^j|\by_i^{\bar{\jmath}(r)},z_i^{(r)},\bG_k^{(r-1)}$ with $z_i^{(r)}=k$ is defined by
\begin{equation}
p(x_i^j|\by_i^{\bar{\jmath}(r)},z_i^{(r)},\bG_k^{(r-1)},\bbet_{kj})=\left\{ \begin{array}{rl}
\phi_1(\frac{x_i^j-\mu_{kj}}{\sigma_{kj}};\tilde{\mu}_i,\tilde{\sigma}^2_i)/\sigma_{kj} & \text{if } 1\leq j\leq c \\
\Phi_1(\frac{b^{\oplus}(x_i^j) - \tilde{\mu}_i}{\tilde{\sigma}_i}) - \Phi_1(\frac{b^{\ominus}(x_i^j) - \tilde{\mu}_i}{\tilde{\sigma}_i}) & \text{otherwise},
\end{array}\right.
\end{equation}
where the real $\tilde{\mu}_{i}=\boldsymbol{\Gamma}_{k}^{(r-1)}[j,\bar{\jmath}]\boldsymbol{\Gamma}_{k}^{(r-1)}[\bar{\jmath}, \bar{\jmath}]^{-1}\by_i^{\bar{\jmath}(r)}$ is the full conditional mean of $y_i^j$, $\boldsymbol{\Gamma}_{k}[j,\bar{\jmath}]$ being the row $j$ of $\boldsymbol{\Gamma}_{k}$ deprived of  element $j$ and $\boldsymbol{\Gamma}_{k}[\bar{\jmath}, \bar{\jmath}]$ being the matrix $\boldsymbol{\Gamma}_{k}$ deprived of the row and the column $j$,  and where $\tilde{\sigma}^2_i$ is the full conditional variance of $y_i^j$ defined by $\tilde{\sigma}^2_i=1 -$ $\boldsymbol{\Gamma}_{k}^{(r-1)}[j,\bar{\jmath}]\boldsymbol{\Gamma}_{k}^{(r-1)}[\bar{\jmath}, \bar{\jmath}]^{-1}\boldsymbol{\Gamma}_{k}^{(r-1)}[\bar{\jmath}, j]$.
As the normalizing constant of \eqref{beta_sampling} is unknown, $\bbet_{kj}^{(r)}$  cannot be directly sampled. This problem is avoided by one iteration of the Metropolis-Hastings algorithm. The instrumental distribution of this Metropolis-Hastings algorithm $q_2(.|\tx,\tz)$ samples a candidate $\bbet_{kj}^{\star}$ according to the posterior distribution of $\bbet_{kj}$ under  conditional independence assumption (this distribution is explicit since the conjugate prior distributions are used). So,
$q_2(.|\tx,\tz)=p(\bbet_{kj}|\tx,\tz,\bG_k=\boldsymbol{I})$. Thus, according to \eqref{beta_sampling}, the candidate $\bbet_{kj}^{\star} $ is accepted with the probability
\begin{equation*}
\rho_2^{(r)}=\min\left\{
\frac{q_2(\bbet_{kj}^{(r-1)}|\tx,\tz)p(\bbet_{kj}^{\star})}{q_2(\bbet_{kj}^{\star}|\tx,\tz)p(\bbet_{kj}^{(r-1)})}
\prod_{\{i: z_i^{(r)}=k\}}
\frac{p(x_i^j|\by_{i}^{\bar{\jmath}(r)},z_i,\bG_k^{(r-1)},\bbet_{kj}^{\star})}
{p(x_i^j| \by_{i}^{\bar{\jmath}(r)},z_i,\bG_k^{(r-1)},\bbet_{kj}^{(r-1)})}
;1 \right\}.
\end{equation*}
Thus, at iteration $(r)$ of Algorithm~\ref{Gibbs}, step \eqref{margins} is performed via one iteration of the following Metropolis-Hastings algorithm whose the stationary distribution is $p(\bbet_{kj}|\bx_{[rk]},\by^{\bar{\jmath}(r)}_{[rk]},\tz,\bbet_{k\bar{\jmath}}^{(r)},\bG_k)$.
\begin{algo}
\begin{align}
\bbet_{kj}^{\star}& \sim q_2(\bbet_{kj}|\tx,\tz)\\
\bbet_{kj}^{(r)}&=\left\{ \begin{array}{rl}
\bbet_{kj}^{\star} & \text{ with probability } \rho_2^{(r)} \\
\bbet_{kj}^{(r-1)}& \text{ with probability } 1-\rho_2^{(r)}.
\end{array}\right.
\end{align}
\end{algo}

Vector $\by_{[rk]}^{j(r)}$ is easily sampled after $\bbet_{kj}^{(r)}$. Indeed, independence between the individuals defines the full conditional distribution of $\by_{[rk]}^j$ by
\begin{equation}
p(\by_{[rk]}^j|\tx,\by^{\bar{\jmath}(r)}_{[rk]},\tz^{(r)},\bbet_{k\bar{\jmath}}^{(r)},\bbet_{kj},\bG_k^{(r-1)})=
\prod_{\{i: z_i^{(r)}=k\}}
p(y_{i}^j|x_i^j,\by_{i}^{\bar{\jmath}(r)},z_i^{(r)},\bbet_{kj},\bG_k^{(r-1)}).
\end{equation}
If $x^j$ is a continuous variable (\emph{i.e.} $1\leq j \leq c$), when $z_i^{(r)}=k$, the full conditional distribution of $y_i^j$ is a Dirac distribution in $\frac{x_i^j-\mu_{kj}^{(r)}}{\sigma_{kj}^{(r)}}$. If $x^j$ is a discrete variable (\emph{i.e.} $c+1\leq j \leq e$), when $z_i^{(r)}=k$, the full conditional distribution of $y_i^j$ is a truncated Gaussian distribution such as,
\begin{equation}
p(y_{i}^j|x_i^j,\by_{i}^{\bar{\jmath}(r)},z_i^{(r)},\bbet_{kj}^{(r)},\bG_k^{(r-1)})=
\frac{\phi_1(y_i^j;\tilde{\mu}_i,\tilde{\sigma}_i^2)}{p(x_i^j;\bbet_{kj}^{(r)})}\mathds{1}_{\{y_i^j \in [b_k^{\ominus (r)}(x_i^j),b_k^{\oplus (r)}(x_i^j)]\}},
 \label{y_sampling}
\end{equation}

So, step \eqref{margins} is performed in two steps. First, $\bbet_{kj}^{(r)}$  is sampled via one iteration of the Metropolis-Hastings algorithm whose the stationary distribution is $p(\bbet_{kj}|\tx,\by^{\bar{\jmath}(r)}_{[rk]},\tz^{(r)},\bbet_{k\bar{\jmath}}^{(r)},\bG_k)$. Second,  $\by_{[rk]}^{j(r)}$ is sampled from \eqref{y_sampling}.

\subsection{Proportion vector sampling}
The aim is to sample from \eqref{proportions}. The sampling from \eqref{proportions} is classical. Indeed, the conjugate Jeffreys non informative prior involves that
\begin{equation}
\bpi|\tz^{(r)} \sim \mathcal{D}_g\left(n_1^{(r)}+\frac{1}{2},\ldots,n_g^{(r)}+\frac{1}{2}\right),
\end{equation}
where $n_k^{(r)} =\sum_{i=1}^n \mathds{1}_{\{z_i^{(r)}=k\}}$.

\subsection{Correlation matrix sampling}
The aim is to sample from \eqref{correlations}. To sample from \eqref{correlations}, we use the approach proposed by \citet{Hof07} in the case of semiparameteric Gaussian copula. First, a covariance matrix is generated by its explicit posterior distribution, and second, the correlation matrix is deduced by normalizing the covariance matrix. As $(\ty,\tz)$ are known in this step, we are in the well-known case of a multivariate Gaussian mixture model with known means. Thus, the sampling according to $\bG_k | \ty^{(r)},\tz^{(r)} $ is performed by the following two steps
\begin{equation}
\boldsymbol{\Lambda}_k| \ty^{(r)},\tz^{(r)} \sim \mathcal{W}^{-1}\left(s_0+n_k^{(r-1)},S_0+ \sum_{ \{i: z_i^{(r)}=k\} } \by_i^{(r)T} \by_i^{(r)}\right), \label{lamda}
\end{equation}
where $\forall 1\leq h,\ell\leq e,\;\bG_k[h,\ell]=\frac{\boldsymbol{\Lambda}_k[h,\ell]}{\sqrt{\boldsymbol{\Lambda}_k[h,h]\boldsymbol{\Lambda}_k[\ell,\ell]}}.$
As the homoscedastic model assumes the equality between the correlation matrices, in such a case we only sample one $\boldsymbol{\Lambda}$ so \eqref{lamda} is replaced by
\begin{equation}
\boldsymbol{\Lambda}| \ty^{(r)},\tz^{(r)} \sim \mathcal{W}^{-1}\left(s_0+n,S_0+ \sum_{i=1}^{n} \by_i^{(r)T} \by_i^{(r)}\right),
\end{equation}
and we put $\boldsymbol{\Lambda}_k=\boldsymbol{\Lambda}$ for $k=1,\ldots,g$.

\end{document}